\documentclass[10pt,twocolumn]{IEEEtran}
\usepackage[all]{xy}
\usepackage{color}
\usepackage{algorithm}
\usepackage{algpseudocode}
\usepackage{subfigure}
\usepackage{bm}
\usepackage{graphicx}
\usepackage{color}
\usepackage{wrapfig}
\usepackage{epsf}
\usepackage{times}
\usepackage{url}
\usepackage{amssymb, amsmath, mathtools, amsthm}
\usepackage{nomencl}
\usepackage{comment}
\makenomenclature
\newtheorem{theorem}{Theorem}

\newtheorem{corollary}{Corollary}
\newtheorem{lemma}{Lemma}
\definecolor{RED}{rgb}{0.6,0.,0.}
\definecolor{BLUE}{rgb}{0.,0.,0.6}
\definecolor{GREEN}{rgb}{0.,0.6,0.}
\definecolor{MALINA}{rgb}{0.6,0.,0.6}
\definecolor{YELLOW}{rgb}{0.8,0.8,0}
\newcommand{\squeezeup}{\vspace{-1.5 mm}}
\begin{document}
\bstctlcite{IEEEexample:BSTcontrol}
\title{Graphical Models in {Meshed} Distribution Grids:
Topology estimation, change detection \& limitations}
\author{
\IEEEauthorblockN{Deepjyoti~Deka\dag, Saurav~Talukdar\ddag, Michael~Chertkov\dag, and Murti~Salapaka \ddag}\\
\IEEEauthorblockA{(\dag) Los Alamos National Laboratory, New Mexico, USA,
(\ddag)University of Minnesota, Minneapolis, USA}
}

\maketitle
\begin{abstract}
Graphical models are a succinct way to represent the structure in probability distributions. This article analyzes the graphical model of nodal voltages in non-radial power distribution grids. Using algebraic and structural properties of graphical models, algorithms exactly determining topology and detecting line changes for distribution grids are presented along with their theoretical limitations. We show that if distribution grids have {cycles/loops of size greater than three}, then nodal voltages are sufficient for efficient topology estimation without additional assumptions on system parameters. In contrast, line failure or change detection using nodal voltages does not require any structural assumption. {Under noisy measurements, we provide the first non-trivial bounds on the maximum noise that the system can tolerate for asymptotically correct topology recovery. The performance of the designed algorithms is validated with non-linear AC power flow samples generated by Matpower on test grids, including scenarios with injection correlations and system noise.}
\end{abstract}

\section{Introduction}
\label{sec:intro}
Power distribution grids comprise part of the power network that include low and medium voltage lines that connect distribution substation to the end users/consumers.The presence of active controllable devices and stochastic resources on the distribution grid and their potential to provide grid services has made estimation problems in distribution grids of paramount importance in recent years. In this paper we discuss issues in topology estimation and change detection in distribution grids. This problem has applications in various areas such as fault detection and localization, and estimation of critical lines that affect locational marginal prices. {Structurally, a majority of distribution grids are designed for radial or tree-like operation \cite{hoffman2006practical}. However, to optimally utilize controllable resources and resiliently deliver power, distribution grids with meshed/loopy operation are being proposed \cite{NY,germany,taiwan}.} In either case, the true operational topology is determined by the current status of breakers/switches on an underlying set of permissible edges as shown in Fig.~\ref{fig:city}. {Such topology reconfiguration \cite{distributiongridbook} may be warranted due to repair work, ensure feasible voltage profiles, or to minimize line losses \cite{hassan2016topology}}. Topology estimation and change detection thus refer to identifying the status and change of status respectively of permissible lines in the grid. In this paper, we synthesize statistical methods for these goals, for the more general case of meshed/loopy distribution grids under noisy measurements. {All algorithms developed in the article extend trivially to the restricted case of radial networks.}

It is worth mentioning that the estimation problems in distribution grids are hampered by the sparsity of real-time meters, including on grid lines. New placement of line/breaker status monitors is further complicated when presence of underground lines are employed. As such we focus on using measurements of nodal voltages for our estimation goals. Such nodal measurements have become more accessible in recent years with usage of installed nodal high fidelity meters such as phasor measurement units (PMUs) \cite{PMU}, micro-PMUs \cite{micropmu}, FNETs \cite{FNET}, and sensors on smart controllable devices. {Our proposed algorithms enable distribution management systems (DMS) to use such measurements to estimate grid connectivity, and validate the correctness of any topology reconfiguration action \cite{distributiongridbook}. Additionally labelled data of known topology reconfiguration available from feeder remote terminal units (FRTU) \cite{distributiongridbook} can be used to estimate statistics of bus voltages and its changes. This can subsequently aid in the selection of thresholds in the related algorithms.}

\subsection{Prior Work}
Estimation problems, including topology estimation and change detection, in power distribution grids has attracted significant attention in recent years. Researchers have looked at multiple approaches, both active and passive, in learning using varying measurement type and availability. Example of such schemes include greedy methods \cite{dekatcns,dekapscc}, voltage signature based methods \cite{arya, berkeley}, probing schemes \cite{cavraro2018graph}, imposing graph cycle constraints \cite{ramstanford} and iterative schemes for addressing missing data \cite{dekasmartgridcomm,sejunpscc}. In contrast to the referred work that employ static voltage samples, learning schemes that exploit correlated voltage measurements from a linear dynamical framework are reported in \cite{sauravacc,sauravacm}.

In work related closest to this article, authors have addressed topology identification using properties of probabilistic graphical model of nodal voltages. \cite{bolognani2013identification} uses signs in inverse covariance matrix of voltage magnitudes for topology identification, but limited to radial topologies in grids with constant $r/x$ (resistance to reactance) line ratio. \cite{dekapscc,dekathreephase} uses voltage conditional independence tests for guaranteed topology identification, but limited to radial distribution grids in single and three-phase networks. \cite{he2011dependency} discusses topology change detection using an approximate graphical model (Markov random field). Under a similar approximate graphical model, topology reconstruction algorithms are proposed for radial and loopy distribution grids in \cite{weng2017distributed} and \cite{liao2016urban} respectively, under independent nodal current injections.

\subsection{Contribution}
In this work, we consider graphical model \cite{wainwright2008graphical} based learning schemes for topology estimation and line failure detection using complex voltage measurements for \emph{meshed/loopy} grids. We show that, under uncorrelated nodal injections, the structure of the voltage graphical model includes additional edges over the topology of the underlying grid graph. Approximate schemes for topology learning have been discussed previously \cite{he2011dependency,weng2017distributed,liao2016urban} by ignoring spurious edges. In contrast, we take a principled theoretical approach and develop two algorithms, and present conditions under which exact topology recovery is possible without ignoring the spurious edges. The first algorithm relies on local neighborhood counting within the estimated graphical model, while the second algorithm uses algebraic sums of terms in the inverse covariance matrix. We show that the first approach is able to estimate the true structure for grids with loops/cycles of size greater than six, while the second approach only requires minimum loop/cycle length to be greater than three (no triangles). It is worth noting that our algorithms do not require knowledge of line impedances, and nodal injection statistics. Next, we develop a topology change detection algorithm that is able to estimate line failures or additions using entries in the graphical model. {Our learning algorithms are proven to be asymptotically correct for noise-less voltage measurements generated with uncorrelated nodal injections. Moreover, we provide non-trivial bounds on the maximum noise in voltage measurements, and maximum injection cross-correlation under which the algorithms are guaranteed to perform accurately. The performance of the developed algorithms, under both noiseless and noisy settings, are calibrated with finite samples from non-linear power flow, using Matpower \cite{matpower}.}

{This article is the journal version of a conference article \cite{dekairep} on learning in bulk grids using DC power flow model. Over \cite{dekairep}, we extend the learning algorithms to linearized AC voltages and analyze the algorithms' computational complexity with detailed proofs not present in the conference version. The algorithm design for topology change detection, and theoretical analysis under noisy measurements and correlated injections are novel contributions of this work. Furthermore, this journal version includes simulations with noisy non-linear AC PF samples for all algorithms (as against DC samples in \cite{dekairep}), that corroborate the novel theoretical development.} The rest of the paper is organized as follows. The next section presents nomenclature and power flow relations. Section \ref{sec:graphicalmodel} discusses properties of graphical model of voltage measurements. Section \ref{sec:neigh} includes the first topology learning algorithm. Section \ref{sec:thres} describes the second learning algorithm. Topology change detection is discussed in Section \ref{sec:change}. Section \ref{sec:noise1} present guarantees on the performance of the learning and change detection algorithms in the presence of measurement noise and injection correlation. Section \ref{sec:simulations} includes simulations results of our work on test cases. Conclusions are included in Section \ref{sec:conclusion}.

\section{Power Grid and Power Flow}
\label{sec:structure}
We consider a power grid and represent its structure by the graph, ${\cal G}=({\cal V},{\cal E})$, where ${\cal V}$ is the set of $N+1$ buses/nodes and ${\cal E}$ is the set of operational undirected lines/edge (see Fig.~\ref{fig:city}). Normally the operational topology is determined by closing switches/breakers within a set of permissible lines, ${\cal E}^{full}$. Such changes may be made hourly or daily depending on the load configuration served and other control needs, hence the need for grid topology estimation and change detection. We denote an edge between two nodes $i$ and $j$ by $(ij)$. Let ${\cal P}^{j}_{i} \equiv \{(ik_1), (k_1k_2),...(k_{n-1}j)\}$ be a set of $n$ distinct undirected edges that connect node $i$ and node $j$. We call ${\cal P}^{j}_{i}$ as a path of length $n$ from $i$ to $j$. If $i = j$ and path ${\cal P}^{j}_{i}$ has length greater than $2$, we term it a `cycle'. The number of edges in the cycle is termed its `cycle length. {The size of the smallest cycle in the graph is called its `minimum cycle length'.} By definition, the minimum cycle length in a radial graph is considered to be infinite as it has no cycles. Note that for general loopy grids, there may be multiple paths between two nodes. The neighbors of a node are the set of nodes it shares an edge with. Nodes with the length of the shortest path connecting them equal to two are termed `two-hop neighbors'. Nodes with degree $1$ are termed `leaves' and their individual neighbors are termed `parents'.
\begin{figure}[!bt]
\centering
\includegraphics[width=0.20\textwidth]{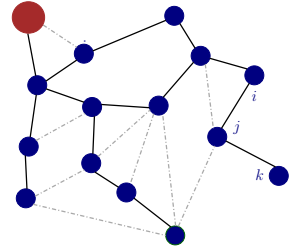}
\caption{Distribution grid with a substation (reference node). Operational edges are denoted by solid lines, while open switches are denoted by dotted lines. ${\cal P}_k^i \equiv \{(kj), (ji)\}$ is one path from $k$ to $i$.
\label{fig:city}}
\end{figure}
For power grid $\cal G$, we use the following power flow model. \\
\textbf{Linear Coupled Power Flow (LC-PF) model:} Let $z_{ij}=r_{ij}+\hat{i} x_{ij}$ denote the complex impedances of line $(ij)$ in the grid ($\hat{i}^2=-1$), where resistance and reactance are denoted by $r_{ij}$ and $x_{ij}$ respectively. The non-linear AC power flow equation for the injection at node $i$ is given by:
\begin{align}
 P_i =p_i+\hat{i} q_i&=\underset{j:(ij)\in{\cal E}}{\sum} V_i(V_i^* - V_j^*)/z_{ij}^*\label{P-complex},\\
&= \underset{j:(ij)\in{\cal E}}{\sum}\frac{v_i^2-v_i v_j\exp(\hat{i}\theta_i-\hat{i}\theta_j)}{z_{ij}^*},\label{P-complex1}
\end{align}
where, the real valued scalars, $v_i$, $\theta_i$, $p_i$ and $q_i$ denote the voltage magnitude, voltage phase, active and reactive power injections respectively. The complex valued voltage and injection are given by $V_i (= v_i\exp(\hat{i}\theta_i))$ and $P_i$ respectively. During normal operation, lossless models for power flow are employed \cite{dekatcns,89BWc,bolognani2016existence} which can be obtained by linearizing the power flow equations in (\ref{P-complex}), assuming small deviations in both phase difference of neighboring nodes ($|\theta_i - \theta_j| <<1$ for edge $(ij)$), and voltage magnitude deviations from the reference bus ($|v_i -1|<<1$ for node $i$).
\begin{align}
 P_i =p_i+\hat{i} q_i&= \underset{j:(ij)\in{\cal E}}{\sum}\frac{(v_i-v_j) -\hat{j}(\theta_i-\theta_j)}{z_{ij}^*},\label{intermediate}
\end{align}
Further one bus (normally the substation bus) is taken as reference and voltages at other buses are measured relative to it. Under the lossless model, the nodal injection at the reference bus is the negative sum of injections at all other nodes. Without a loss of generality, the reference bus is ignored and the power flow equations are restricted to the $N$ non-reference buses in the grid. Using $v,\theta, p, q$ to represent the vector of nodal voltage magnitude, phase, active and reactive injections respectively, we get the following Linear Coupled Power Flow (LC-PF) model, in matrix form:
\begin{align}
&\begin{bmatrix}v \\\theta\end{bmatrix}= H_{(g,\beta)}^{-1}\begin{bmatrix}p \\q\end{bmatrix}, \text{where~} H_{(g,\beta)}=\begin{bmatrix}H_g &H_{\beta}\\ H_{\beta} &-H_{g}\end{bmatrix} \label{PF_LPV_p}
\end{align}
Here $H_{\beta}$ ($H_{g}$) is the reduced weighted Laplacian matrix for the grid ${\cal G}$ with edge weights given by susceptances $\beta$ (conductances $g$). Here $g_{ij}+\hat{i}\beta_{ij}$ equals $\frac{1}{r_{ij}-\hat{i}x_{ij}}$. The sparsity of $H_\beta$ and $H_g$ encodes the grid structure as noted below:
\begin{align}
 H_g(i,j) = \begin{cases}\sum_{k:(ik) \in {\cal E}}g_{ik} ~~\text{if~} i=j\\
-g_{ij} ~~\text{if~} (ij) \in {\cal E}\\
 0 ~~\text{otherwise}\end{cases}.\label{Laplacian}
\end{align}
The reduction is derived by removing the row and column corresponding to the reference bus from the weighted Laplacian matrix. We use $v, \theta, p, q$ for the $N$ dimensional reduced vectors to denote the voltage magnitude, phase, active and reactive injections at the non-substation nodes respectively, from this point onward. Note that (\ref{PF_LPV_p}) is a generalization of the LinDistFlow model \cite{89BWc,89BWb,89BWa} used in radial networks to loopy grids. Further ignoring voltage magnitude deviations converts the model to the standard DC power flow model \cite{abur2004power} used in transmission grids. The accuracy of the LC-PF voltages with non-linear power flow samples are discussed in \cite{dekathreephase,bolognani2016existence}. In the next section, we discuss the properties of the distribution of voltage fluctuations under the invertible LC-PF model.

\section{Probabilistic Graphical Model of Complex Voltages}
\label{sec:graphicalmodel}
For a $N$ bus (ignoring the reference bus) system, the total number of scalar voltage variables is $2\times N$, considering both magnitude and phase. We use graphical models to represent the structure within the distribution of voltages under ambient fluctuations of nodal injections. The following assumption states the model of nodal injections at the non-reference nodes.

\textbf{Assumption $1$}: Fluctuations at non-reference nodal injections in the grid are uncorrelated zero-mean Gaussian random variables with non-zero covariances. Thus, $(p_i,q_i)$ is uncorrelated from $(p_j,q_j)$ if $i\not =j$, while each is a two-dimensional zero-mean Gaussian.

This assumption, similar to prior work in literature {\cite{dekatcns,he2011dependency,liao2016urban,cavraro2017voltage} arises from the fact that fluctuations in loads in the ambient regime are typically uncorrelated. If small linear trends in the normal injections and consequently voltages exist, such trends can be empirically removed by taking the difference of consecutive measurements to obtain uncorrelated fluctuations \cite{sauravacc,liao2016urban}.} Note that active and reactive injections at the same node may be correlated, under Assumption $1$. {In later sections, we present results on the effect of inter-nodal injection correlation on the performance of the developed learning algorithms.}

{Our algorithms for topology identification and change detection need only the inverse covariance matrix of voltage fluctuations, and are agnostic to the underlying parametric distributions. The Gaussianity assumption of injection fluctuations and thereby of voltages is used only for efficient estimation of the inverse covariance under low number of measurement samples. Existing work \cite{liao2016urban,cavraro2017voltage,yury,bienstock2014chance} show that real voltage and injection data indeed follow Gaussian distributions.}

Here, the injection vector \text{\footnotesize$\begin{bmatrix}p\\q\end{bmatrix}$} is modelled by the Gaussian random variable ${\cal P}^{LC}(p, q)\equiv \mathcal{N}(0, \Sigma_{(p,q)})$ where the covariance matrix of injections is given by
\begin{align}
\squeezeup
\Sigma_{(p,q)}= \begin{bmatrix}\Sigma_{pp} &\Sigma_{pq} \\ \Sigma_{qp} &\Sigma_{qq}\end{bmatrix} = \mathbb{E}\left[\begin{bmatrix}p\\q\end{bmatrix}[p^T~ q^T]\right] .\label{sigma_lc}
\end{align}
Under Assumption $1$, each block in $\Sigma_{(p,q)}$ is a diagonal matrix. Given the voltages \text{\footnotesize$\begin{bmatrix}v\\\theta\end{bmatrix}$} are related by a linear model to injections, their distribution is also a zero-mean Gaussian $\mathcal{N}(0, \Sigma_{(v,\theta)})$ \cite{gubner2006probability}. Using (\ref{PF_LPV_p}), the covariance matrix of voltages is given by
\begin{align}
\Sigma_{(v,\theta)}= \begin{bmatrix}\Sigma_{vv} &\Sigma_{v\theta} \\ \Sigma_{\theta v} &\Sigma_{\theta\theta}\end{bmatrix} = H^{-1}_{(g,\beta)}\Sigma_{(p,q)}H^{-1}_{(g,\beta)}. \label{covar_varepsilon}
\end{align}
The following result describes the analytic form of ${\Sigma^{-1}_{(v,\theta)}}$ which will later be used for structure estimation.
\begin{lemma}\label{inverse_LC_PF}
For LC-PF, the inverse covariance matrix ${\Sigma^{-1}_{(v,\theta)}}$ of nodal voltages satisfies
\begin{align*}
&{\Sigma^{-1}_{(v,\theta)}} = \begin{bmatrix} J_{vv} &J_{v\theta}\\J_{\theta v} &J_{\theta\theta}\end{bmatrix}\text{~where~}\\
&J_{vv}=\text{\small $H_{g}D^{-1}(\Sigma_{qq}H_{g}-\Sigma_{pq}H_{\beta})-H_{\beta}D^{-1}(\Sigma_{pq} H_{g}-\Sigma_{pp}H_{\beta})$},\nonumber\\
&J_{v\theta}=\text{\small $H_{g}D^{-1}(\Sigma_{qq}H_{\beta} +\Sigma_{pq}H_{g})-H_{\beta}D^{-1}(\Sigma_{pq}H_{\beta}+\Sigma_{pp}H_{g})$},\nonumber\\
&J_{\theta v}=\text{\small $H_{\beta}D^{-1}(\Sigma_{qq}H_{g} -\Sigma_{pq}H_{\beta})+H_{g}D^{-1}(\Sigma_{pq} H_{g}-\Sigma_{pp}H_{\beta})$},\nonumber\\
&J_{\theta\theta}=\text{\small $H_{\beta}D^{-1}(\Sigma_{qq}H_{\beta} +\Sigma_{pq}H_{g})+H_{g}D^{-1}(\Sigma_{pq}H_{\beta}+\Sigma_{pp}H_{g})$},\nonumber\\
&D(i,i) = |\Sigma_{pp}(i,i)\Sigma_{qq}(i,i)-{\Sigma_{pq}}^2(i,i)| \text{~for diagonal~}D.
\end{align*}
\end{lemma}
\begin{proof}
In (\ref{covar_varepsilon}), each block in $\Sigma_{(p,q)}$ is a diagonal matrix with $\Sigma_{pq} = \Sigma_{qp}$. Thus the following holds
\begin{align}
\begin{bmatrix}\Sigma_{pp} &\Sigma_{pq} \\ \Sigma_{pq} &\Sigma_{qq}\end{bmatrix}^{-1} = \begin{bmatrix}D^{-1}\Sigma_{qq} &-D^{-1}\Sigma_{pq}\\ -D^{-1}\Sigma_{pq} &D^{-1}\Sigma_{pp}\end{bmatrix},\nonumber
\end{align}
where, $D$ is a diagonal matrix with {the $i^{th}$ diagonal entry} $D(i,i)$ given by the determinant of $\begin{bmatrix}\Sigma_{pp}(i,i) &\Sigma_{pq}(i,i) \\\Sigma_{pq}(i,i) &\Sigma_{qq}(i,i)\end{bmatrix}$. Inverting $\Sigma_{(v,\theta)}$ using (\ref{covar_varepsilon}) then proves the result.
\end{proof}
Next we describe the graphical model of the distribution of nodal voltages.\\
\textbf{Graphical Model}: By definition, the probability distribution of a $n$ dimensional random vector $X = [X_1, X_2,..X_n]^T$ corresponds to an undirected graphical model ${\cal GM}$ \cite{wainwright2008graphical} with vertex set ${\cal V_{GM}}$ representing variables and edges representing conditional dependence. For node $i$ in $\mathcal{GM}$, its neighbors form the smallest set of nodes $N(i) \subset {\cal V_{GM}}- \{i\}$ such that for any node $j \not\in N(i)$, $i$ is conditionally independent of $j$ given the set $N(i)$, i.e., ${\cal P}(X_i|X_ {N(i)},X_j) = {\cal P}(X_i|X_ {N(i)})$.

For a Gaussian graphical model, it is known that the edges in the graphical model correspond to non-zero terms in the inverse covariance matrix (also called `concentration' matrix) \cite{wainwright2008graphical}. In our case, we determine the structure of the graphical model $\mathcal{GM}$ of voltages using properties of $\Sigma_{(v,\theta)}^{-1}$. Note that there are twice as many nodes in the graphical model $\mathcal{GM}$ as there are buses in the grid as $\mathcal{GM}$ includes separate nodes for bus voltage magnitudes and phases.
\begin{theorem}\label{structure_LC_PF}
The graphical model $\mathcal{GM}$ for nodal voltage magnitudes and phase angles in grid $\cal G$ includes edges between voltage magnitudes and phase angles only at the same bus, neighboring buses, and two-hop neighboring buses.
\end{theorem}
\begin{proof}
Edges in $\mathcal{GM}$ correspond to non-zero terms in ${\Sigma^{-1}_{(v,\theta)}}$ with analytic form given in Lemma \ref{inverse_LC_PF}. We prove the statement first by showing that for voltage magnitude and/or phases at buses $i$ and $j$ three or more hops away, the corresponding entries in ${\Sigma^{-1}_{(v,\theta)}}$ are zero for each block $J_{vv}$, $J_{v\theta}$, $J_{\theta v}$ and $J_{\theta\theta}$. First consider the four terms in the expression for $J_{v\theta}$ in Lemma \ref{inverse_LC_PF}. While $D^{-1}, \Sigma_{qq}, \Sigma_{pq}, \Sigma_{pp}$ are all diagonal matrices, matrices $H_g, H_\beta$ have non-zero values for diagonal terms and for neighboring nodes (see (\ref{Laplacian})). From direct multiplication it is clear that $J_{vv} (i,j) = 0$ if $i,j$ are not neighbors and do not have common neighbor. Similarly it follows for $J_{v\theta}$, $J_{\theta v}$ and $J_{\theta\theta}$, the statement follows.

Note that if $(ij) \in \mathcal{E}$ or node $k$ exists such that $(ik), (kj) \in \mathcal{E}$, then the corresponding entry ${\Sigma^{-1}_{(v,\theta)}}$ is non-zero unless for pathological cases that form a set of measure zero. Thus the statement holds.
\end{proof}
Remark: In the remaining part of the manuscript, we will ignore such pathological cases that induce degeneracy; thus if there is a true link in $\mathcal{G}$ between nodes $i$ and $j$ then terms in $J_{vv},\ J_{v\theta},\ J_{\theta,\theta}$, and $J_{\theta v}$ do not conspire to be zero.
\begin{figure}[tb]
\centering\hfill
\subfigure[]{\includegraphics[width=0.25\columnwidth]{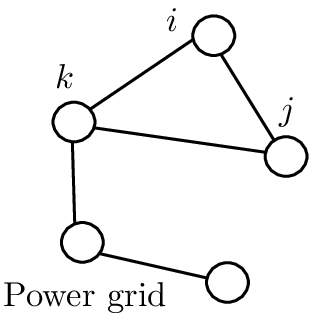}\label{fig:grid}}\hfill
\subfigure[]{\includegraphics[width=0.36\columnwidth]{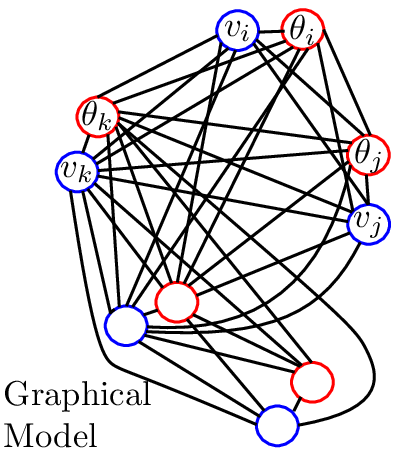}\label{fig:graphical_LC}}\hfill
\subfigure[]{\includegraphics[width=0.27\columnwidth]{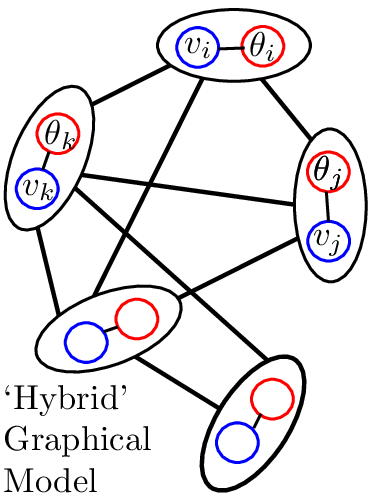}\label{fig:graphical_hybrid}}\hfill
\caption{(a) Loopy Power grid graph ${\mathcal{G}}$ (b) Graphical model $\mathcal{GM}$ for complex voltages (c) `Hybrid' graph $\mathcal{GM}_{hd}$ by combining voltage magnitude and phase nodes for same bus in graphical model. \label{fig:graphical_model}}
\end{figure}
Fig.~\ref{fig:graphical_model} depicts the structure of a loopy grid $\mathcal{G}$ and its associated graphical model $\mathcal{GM}$ of complex voltages. Note that if nodes pertaining to voltage magnitude and phase angles are combined into a `hybrid' node, we get a similar sized network $\mathcal{GM}_{hd}$ as the original grid but with additional edges between two hops neighbors as shown in Fig.~\ref{fig:graphical_hybrid}. To complete this section, we describe the estimation of $\Sigma^{-1}_{(v,\theta)}$ from $i.i.d$ complex voltage samples.

\subsection{Estimation of Voltage Graphical Model}
\label{sec:graphicalmodellearning}
Consider i.i.d. samples of the measured vector $y^k = \begin{bmatrix}v^k\\\theta^k\end{bmatrix}$ for $1\leq k\leq n$ in the grid. If $n$ is large, we use direct inverse of the sample covariance matrix. {Note that such inversion can be performed for general injection distributions as well.} However for scenarios with restricted number of samples, we use the maximum likelihood estimator of a Gaussian graphical model \cite{wainwright2008graphical} with a sparsity constraint to get the estimate ${\hat{\Sigma}^{-1}_{(v,\theta)}}$ using
\begin{align}
{\hat{\Sigma}^{-1}_{(v,\theta)}} = \arg \min_{S} \{-\log\det S +\langle S,\hat{\Sigma}_{(v,\theta)}\rangle +\lambda\|S\|_1\} \nonumber\\
\text{where~}\hat{\Sigma}_{(v,\theta)} = \frac{1}{n-1}\sum^n_{k=1} (y^k - \sum^n_{k=1}\frac{y^k}{n})(y^k - \sum^n_{k=1} \frac{y^k}{n})^T.\nonumber
\end{align}
This optimization problem is termed \textbf{Graphical Lasso} \cite{yuan2007model,tibshirani} in literature. The computational complexity for it scales as $O(N^3)$ where $N$ is the number of variables in the system \cite{tibshirani}.

Using Theorem \ref{structure_LC_PF} we present the first topology learning algorithm in the next section using separation rules.

\section{Topology Estimation based on Local Neighborhoods}
\label{sec:neigh}
The following result presents separation rules in $\mathcal{GM}_{hd}$ that enable identification of true edges between non-leaf nodes in the grid graph $\mathcal{G}$.
\begin{theorem}\label{condind_loopy}
Let the graphical model $\mathcal{GM}$ for voltages in power grid $\mathcal{G}$ be estimated and the `hybrid' graph $\mathcal{GM}_{hd}$ constructed by merging voltages and phase angles at the same bus. Let minimum cycle length in $\mathcal{G}$ be greater than $6$. Consider edge (ij) in $\mathcal{G}\mathcal{M}_{hd}$. Then, $\mathcal{GM}_{hd}$ has nodes $k,l$ with paths $k-i-l$ and $k-j-l$ but $(kl)$ not present \textbf{iff} $(ij)$ is a true edge in $\cal G$ between non-leaf nodes $i$ and $j$.
\end{theorem}
The proof is presented in Appendix \ref{proof_condind}. While the above theorem enables identification of non-leaf nodes and edges between them in $\mathcal{G}$, we still need to find edges that connect leaves to their true parent. The following result enables identification of true edges associated with leaf nodes, with proof given in Appendix \ref{proof_condind_leaf}.
\begin{theorem}\label{condind_leaf}
Let the graphical model $\mathcal{GM}$ for voltages in power grid $\mathcal{G}$ be estimated and the `hybrid' graph $\mathcal{GM}_{hd}$ constructed by merging voltages and phase angles at the same bus. Let minimum cycle length in $\mathcal{G}$ be greater than $6$ and number of non-leaf nodes be at least $3$. Consider non-leaf node $i$ and leaf node $j$ in $\mathcal{G}$ such that $(ij)$ is in $\mathcal{GM}_{hd}$. Then, $(ij)$ is a true edge in $\cal G$ \textbf{iff} non-leaf neighbors of $i$ in ${\cal G}$ and non-leaf neighbors of node $j$ in $\mathcal{GM}_{hd}$ are the same.
\end{theorem}

\begin{figure}[bt]
\centering
\includegraphics[width=0.4\textwidth]{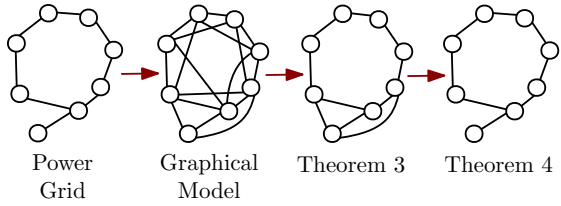}
\vspace{-.25cm}
\caption{Steps in Algorithm $1$. The power grid on the left gives rise to the graphical model with edges between nodes two hops or less away. First Theorem \ref{condind_loopy} is used to identify the true edges between non-leaf nodes and then Theorem \ref{condind_leaf} is used to identify the parents of leaf nodes.}
\label{fig:example}
\vspace{-3mm}
\end{figure}
We now have the necessary tools to develop topology learning algorithm for loopy grids provided the assumptions of minimum cycle length being greater than $6$, and number of non-leaf nodes greater than $2$ hold. {Note that the first condition is necessary for identification of edges between non-leaf nodes while the second is needed to identify neighbors of leaf nodes. If cycle length is $6$ or less, nodes separated by $3$ hops in the grid graph can have multiple shortest paths between them. This prevents separation between true and spurious edges in the corresponding graphical model. These assumptions are not restrictive in general as distribution grids, if loopy, have generally large loops (weakly meshed) and furthermore many (greater than $2$) non-leaf nodes \cite{NY,taiwan}.} The steps of learning the grid from the estimated graphical model $\mathcal{GM}_{hd}$ are listed in Algorithm $1$. An example for the learning steps is depicted in Fig.~\ref{fig:example}.

\begin{algorithm}
\caption{Topology Learning: Neighborhood Search}
\textbf{Input:} Inverse covariance matrix of nodal voltages,\\ $~~~~~~~~~~\Sigma^{-1}_{(v,\theta)} =$ \text{\small$\begin{bmatrix} J_{vv} &J_{v\theta}\\J_{\theta v}&J_{\theta\theta}\end{bmatrix}$}, threshold $\tau_1 >0$ \\
\textbf{Output:} Grid $\cal G$ \\
\begin{algorithmic}[1]
\State Construct hybrid graphical model $\mathcal{GM}_{hd}$ with edges $(ij)$ if $|J_{vv}(i,j)| >\tau_1$ for buses $i\neq j$. \label{alg_1_step1}
\ForAll{Edge $(ij) \in \mathcal{GM}_{hd}$}
\If{Nodes $k,l$ exist satisfying Theorem \ref{condind_loopy}}
\State Add nodes $i,j$ as non-leaf nodes with true edge $(ij)$ in $\mathcal{G}$
\EndIf
\EndFor
\State Add unmarked nodes as leaves in $\mathcal{G}$.
\ForAll{Leaf node $j$ and non-leaf node $i$ in $\mathcal{G}$}
\If{Edge $(ij) \in \mathcal{GM}_{hd}$ and $i,j$ satisfy Theorem \ref{condind_leaf}}
\State Add edge $(ij)$ in $\mathcal{G}$.
\EndIf
\EndFor
\end{algorithmic}
\end{algorithm}
{\textbf{Thresholding for Empirical Estimates:} Note that empirical estimate of voltage covariance may differ from the true value due to finite samples. To address the finiteness of data, we introduce thresholds $\tau_1>0$ in Step \ref{alg_1_step1} of Algorithm $1$ to infer existence of edge $(ij)$ in the hybrid $\mathcal{GM}_{hd}$. {Similar thresholds are introduced in subsequent algorithms in this article as well. Historical training data, including that from topology reconfiguration events, can be used for identifying optimal thresholds. Threshold selection depends on the noise levels as theoretically analyzed in Section \ref{sec:noise1}. We demonstrate the effect of threshold selection on algorithm's performance in Section \ref{sec:simulations}.}

\textbf{Computational Complexity:} Considering general grids without any restriction on nodal degree, determining edges between non-leaf nodes takes $O(N^4)$
comparisons between each candidate edge $(ij)$ and their neighbor sets. Determination of edges between leaves and their parents is $O(N^3)$ and thus the overall complexity is $O(N^4)$ in the worst case.

It is worth mentioning that prior work on learning radial distribution grids exactly using graphical models \cite{dekapscc} cannot be directly applied here as presence of multiple paths in loopy grids render conditional independence tests ineffective. On the other hand, approximate algorithms such as in \cite{liao2016urban} assume, based on realistic data, that the edges between two-hop neighbors in the graphical model are weaker than direct links and hence can be omitted. Our effort in Algorithm $1$ is to develop a theoretical learning approach where such assumptions are not needed, though it has topological restriction. In the next section, we move beyond topological rules and design an improved algorithm based on algebraic properties of the inverse covariance matrix.

\section{Topology Estimation using Sign Rules}
\label{sec:thres}
As discussed in the previous section, extracting the true grid topology $\mathcal{G}$ from the estimated graphical model $\mathcal{GM}$ or $\mathcal{GM}_{hd}$ requires separating the true edges from ones between two hop neighbors. The following result gives a sign based separation test between such edges.
\begin{theorem}\label{threshld_loopy}
Consider the graphical model $\mathcal{GM}$ of nodal voltages with inverse covariance matrix ${\Sigma^{-1}_{(v,\theta)}} = \begin{bmatrix} J_{vv} &J_{v\theta}\\J_{\theta v} &J_{\theta\theta}\end{bmatrix}$ in grid $\cal G$ with minimum cycle length greater than three. $J_{vv}(i,j) + J_{\theta\theta}(i,j) < 0$ if and only if $(ij)$ is a true edge in $\mathcal{G}$.
\end{theorem}
The proof is provided in Appendix \ref{proof_threshld_loopy}. Note that the minimum cycle length being greater than $3$ is necessary to prove Theorem \ref{threshld_loopy}. However absence of triangles (three node cycles) is not a restrictive assumption as urban networks indeed have grid-like layout. It is worth mentioning that a similar sign based rule for voltage magnitudes was proposed in learning radial distribution grids with constant $r/x$ transmission lines in \cite{bolognani2013identification}. Using nodal voltage magnitudes and phase angles, our method is thus able to generalize that on both counts: (a) loopy grid learning (b) variable $r/x$ line parameters. We summarize the steps in using signs for topology estimation in Algorithm $2$. {To account for finite sample errors in estimated inverse covariance matrix, we introduce threshold $\tau_2>0$ in Step \ref{alg_2_step1} of Algorithm $2$ to determine if $(i,j)$ is an edge.}
\begin{algorithm}
\caption{Topology Learning: Sign Rule}
\textbf{Input:} Inverse covariance matrix of nodal voltages,\\ $~~~~~~~~~~\Sigma^{-1}_{(v,\theta)} =$ \text{\small$\begin{bmatrix} J_{vv} &J_{v\theta}\\J_{\theta v}&J_{\theta\theta}\end{bmatrix}$}, threshold $\tau_2 >0$ \\
\textbf{Output:} Grid $\cal G$
\begin{algorithmic}[1]
\ForAll{buses $i,j \in {\cal G}$}
\If{$J_{vv}(i,j)+J_{\theta\theta}(i,j) < -\tau_2$} \label{alg_2_step1}
\State Insert $(ij)$ in $\cal G$
\EndIf
\EndFor
\end{algorithmic}
\end{algorithm}

\textbf{Computational Complexity:} Determining the edges from the entries of the inverse covariance matrix of voltages takes $O(N^2)$ operations as all possible edge pairs are checked for existence.

Notice that both Algorithm $1$ and Algorithm $2$ use only voltage measurements to estimate the topology. In case additionally statistics of injections at the nodes is also available, both the topology as well as values of line impedances can be determined as described in the next result.
\begin{lemma}\label{thm:completion}
In grid $\cal G$ with voltage covariance matrix $\Sigma_{(v,\theta)}$ and injection covariance matrix $\Sigma_{(p,q)}$, matrices $H_g$ and $H_\beta$ are given by the following
\begin{align}
H_{(g,\beta)} =\begin{bmatrix}H_g &H_{\beta}\\ H_{\beta} &-H_{g}\end{bmatrix} = {\Sigma^{1/2}_{(p,q)}}\sqrt{{\Sigma^{-1/2}_{(p,q)}}{\Sigma^{-1}_{(v,\theta)}}{\Sigma^{-1/2}_{(p,q)}}}{\Sigma^{1/2}_{(p,q)}}.\nonumber
\end{align}
\end{lemma}
\begin{proof}
Using Eq.~(\ref{covar_varepsilon}), it is clear that ${\Sigma^{-1/2}_{(p,q)}}{\Sigma^{-1}_{(v,\theta)}}{\Sigma^{-1/2}_{(p,q)}}$ is a positive definite matrix and hence has a unique square root $\Sigma^{-1/2}_{(p,q)}H_{(g,\beta)}\Sigma^{-1/2}_{(p,q)}$. The result thus follows.
\end{proof}
In the next section, we look at a related problem of topology change detection, instead of topology estimation and discuss properties of voltage graphical models that enable detection.
\section{Topology Change Detection}
\label{sec:change}
Topology change detection refers to the problem of identifying changes to the topology such as line failures. In radial distribution grids, line failure can be immediately estimated given by the sub-network that loses power. In loopy girds, due to presence of multiple paths, estimation based on connectivity loss may not be possible. In the next result, we consider a single change in topology (line addition or failure) and show that such an event can be detected by comparing the change in the inverse covariance matrix of voltage fluctuations before and after the event.
\begin{theorem}\label{res:detection}
For grid $\mathcal{G}$ with inverse covariance matrix of voltage measurements ${\Sigma^{-1}_{(v,\theta)}} = \begin{bmatrix} J_{vv} &J_{v\theta}\\J_{\theta v} &J_{\theta\theta}\end{bmatrix}$, let $\Delta J_{vv}$ be the change in sub-matrix $J_{vv}$ after a single line event (addition or removal). Similarly define $\Delta J_{\theta\theta}$ Then the following hold:\\
\begin{align}
\begin{array}{l}
\Delta J_{vv}(i,i)+\Delta J_{\theta\theta}(i,i),\\
\Delta J_{vv}(j,j)+\Delta J_{\theta\theta}(j,j)
\end{array} & = \begin{cases} 0 \text{~no change in edge (ij)}\\
>0 \text{~edge (ij) is added}\\
<0 \text{~edge (ij) is removed}\end{cases}.\nonumber
\end{align}
\end{theorem}\begin{proof}
Let the neighborhood of node $i$ be $\mathcal{N}_i$. From Lemma \ref{inverse_LC_PF}, it is clear that the expression for the diagonal of $J_{vv}$ follows:
 \begin{align}
 &J_{vv}(i,i)+J_{\theta\theta}(i,i) =\nonumber\\
 &\smashoperator[lr]{\sum_{k \in \mathcal{N}_i\cup\{i\}}}D^{-1}(k,k)(\Sigma_{qq}(k,k)+\Sigma_{pp}(k,k))(H^2_g(i,k)+ H^2_\beta(i,k))\nonumber
 \end{align}
Note that this stays the same if the neighbor set $\mathcal{N}_i$ stays constant. Thus for nodes $i,j$, if the edge status doesn't change, $ J_{vv}(i,i)+J_{\theta\theta}(i,i)$ and $ J_{vv}(j,j)+J_{\theta\theta}(j,j)$ equal zero. Next consider the case where edge $(ij)$ is added to the network. Comparing before and after the event, we have
 \begin{align}
 &\Delta J_{vv}(i,i)+\Delta J_{\theta\theta}(i,i) =\nonumber\\
 ~&2D^{-1}(i,i)(\Sigma_{qq}(i,i)+\Sigma_{pp}(i,i))\smashoperator[lr]{\sum_{k \in \mathcal{N}_i -\{j\}}}(g_{ik}g_{ij}+ \beta_{ik}\beta_{ij})\nonumber\\
 ~+&D^{-1}(j,j)(\Sigma_{qq}(j,j)+\Sigma_{pp}(j,j))(g^2_{ij}+\beta^2_{ij})\nonumber\\
 ~+&D^{-1}(i,i)(\Sigma_{qq}(i,i)+\Sigma_{pp}(i,i))(g^2_{ij}+\beta^2_{ij}) > 0 \label{addtn}
 \end{align}
Similarly, we can show that $\Delta J_{vv}(j,j)+\Delta J_{\theta\theta}(j,j) >0$. For edge $(ij)$ removal, it follows similarly that the opposite sign in (\ref{addtn}) is derived. Hence the result holds.\end{proof}
The steps in topology change detection are listed in Algorithm $3$. {Similar to Algorithms $1$ and $2$, we introduce a threshold $\tau_3 >0$ in Step \ref{alg_3_step_1} of Algorithm $3$ to account for empirical errors while identifying the terminal nodes of the changed edge.}
\begin{algorithm}
\caption{Topology Change Detection}
\textbf{Input:} Inverse covariance matrix of nodal voltages,\\ $~\Sigma^{-1}_{(v,\theta)} =$ \text{\small$\begin{bmatrix} J_{vv} &J_{v\theta}\\J_{\theta v}&J_{\theta\theta}\end{bmatrix}$}
 before and after event, threshold $\tau_3 >0$\\
\textbf{Output:} Edge added/removed\\
\begin{algorithmic}[1]
\ForAll{buses $i \in {\cal G}$}
\If{$|\Delta J_{vv}(i,i)+\Delta J_{\theta\theta}(i,i)| > \tau_3$} \label{alg_3_step_1}
\State Mark $i$ as terminal node of changed edge.
\EndIf
\EndFor
\State For terminal nodes $i,j$, edge $(ij)$ is added if $\Delta J_{vv}(i,i)+\Delta J_{\theta\theta}(i,i) > 0$, else it is removed.
\end{algorithmic}
\end{algorithm}

\textbf{Computational Complexity:} Determining the possible changed edge takes $O(N)$ operations as it entails only checking the changes in the diagonal entries in the inverse covariance matrix.

Note that unlike the topology learning algorithms, topology change detection does not need any assumption of minimum cycle length. It is worth mentioning that in practical settings, line failure/change would entail dynamics in the nodal voltage measurements and other time-series techniques may be used for their detection. Our theoretical result is primarily directed at understanding if stable pre- and post event voltage graphical models can be used for change detection.
\section{Effect of noise and injection correlation}\label{sec:noise1}
{The designed Algorithms $1, 2$ and $3$ assume noiseless measurements of complex nodal voltages. {However, practical smart meter and PMU data are corrupted with measurement noise \cite{cavraro2017voltage,ami_noise,pmu_noise} with variance ranging between $.1\%$ to $.5\%$ of the measurement values.} We now analyze guarantees on their performance in the presence of zero-mean Gaussian measurement noise $(n_v,n_\theta)$ of non-zero covariance $\Sigma_{(n_v,n_\theta)}$, that is uncorrelated from the underlying true voltages. Let the voltage measurement with noise be denoted by \text{\footnotesize$\begin{bmatrix}\Tilde{v}\\\Tilde{\theta}\end{bmatrix} = \begin{bmatrix}v\\\theta\end{bmatrix} + \begin{bmatrix}n_v\\n_\theta\end{bmatrix}$}. By uncorrelation between true voltages and noise, the covariance matrix of the observed measurements is given by:
\begin{align}
 \Sigma_{(\Tilde{v},\Tilde{\theta})} &= \Sigma_{(v,\theta)} + \Sigma_{(n_v,n_\theta)},\text{where,}\label{deviatesigma}\\ \Sigma_{(n_v,n_\theta)} &=\begin{bmatrix}\Sigma_{n_vn_v} &\Sigma_{n_v n_\theta} \\ \Sigma_{n_\theta n_v} &\Sigma_{n_\theta n_\theta}\end{bmatrix}
 \text{and~~} \Sigma^{-1}_{(\Tilde{v},\Tilde{\theta})} = \begin{bmatrix} J_{\Tilde{v}\Tilde{v}} &J_{\Tilde{v}\Tilde{\theta}}\\J_{\Tilde{\theta}\Tilde{v}} &J_{\Tilde{\theta}\Tilde{\theta}}\end{bmatrix}\label{deviatesigma1}
\end{align}
Using the Woodbury formula \cite{hager1989updating} for matrix inversion, we have
\begin{align}
 \Delta\Sigma^{-1}_{(v,\theta)}&= \Sigma^{-1}_{(\Tilde{v},\Tilde{\theta})} - \Sigma^{-1}_{(v,\theta)}\nonumber\\ &= -\Sigma^{-1}_{(v,\theta)}(\Sigma^{-1}_{(n_v,n_\theta)} + \Sigma^{-1}_{(v,\theta)})^{-1}\Sigma^{-1}_{(v,\theta)}, \label{deviateinverse}
\end{align}
where $\Delta\Sigma^{-1}_{(v,\theta)}$ is the deviation in the inverse covariance matrix due to noise.
The following theorem provides an upper bound on magnitude of entries in $\Delta\Sigma^{-1}_{(v,\theta)}$.
\begin{lemma}\label{noise_bound}
Consider $\Delta\Sigma^{-1}_{(v,\theta)}=\Sigma^{-1}_{(\Tilde{v},\Tilde{\theta})}-\Sigma^{-1}_{(v,\theta)}$ due to noise of non-zero covariance $\Sigma_{(n_v,n_\theta)}$ in voltage measurements. Let $\lambda^{max}_X, \lambda^{min}_X$ denote the maximum and minimum eigenvalues respectively in square matrix $X$. Then
$$\max_{i,j} |\Delta\Sigma^{-1}_{(v,\theta)}(i,j)| \leq \frac{\lambda^{max}_{\Sigma_{(n_v,n_\theta)}}(\lambda^{max}_{H^2_{(g,\beta)}})^2}{(\lambda^{min}_{\Sigma_{(p,q)}})^2}$$
\end{lemma}
\begin{proof}Let $e_i$ be the basis vector in $\mathbb{R}^{2N}$ with $1$ at the $i^{th}$ entry and zero elsewhere. From Eq.~\ref{deviateinverse}, $\Delta\Sigma^{-1}_{(v,\theta)}$ is a symmetric and negative definite matrix. We have
\begin{align*}
(e^T_i - e^T_j)\Delta\Sigma^{-1}_{(v,\theta)}(e_i - e_j)&\leq 0,~-e^T_i\Delta\Sigma^{-1}_{(v,\theta)}e_i\leq \lambda^{max}_{-\Delta\Sigma^{-1}_{(v,\theta)}},\nonumber\\
(e^T_i + e^T_j)\Delta\Sigma^{-1}_{(v,\theta)}(e_i + e_j)&< 0,-e^T_j\Delta\Sigma^{-1}_{(v,\theta)}e_j\leq \lambda^{max}_{-\Delta\Sigma^{-1}_{(v,\theta)}}\nonumber
\end{align*}
Using that, we have,
\begin{align}
\max_{i,j} |\Delta\Sigma^{-1}_{(v,\theta)}(i,j)| &\leq \lambda^{max}_{-\Delta\Sigma^{-1}_{(v,\theta)}}\nonumber\\
&\leq (\lambda^{max}_{\Sigma^{-1}_{(v,\theta)}})^2/\lambda^{min}_{(\Sigma^{-1}_{(n_v,n_\theta)} + \Sigma^{-1}_{(v,\theta)})}\label{n_1}\\
&\leq (\lambda^{max}_{H^2_{(g,\beta)}}\lambda^{max}_{\Sigma^{-1}_{(p,q)}})^2/\lambda^{min}_{\Sigma^{-1}_{(n_v,n_\theta)}}\label{n_2}\\
&\leq \lambda^{max}_{\Sigma_{(n_v,n_\theta)}}(\lambda^{max}_{H^2_{(g,\beta)}}/\lambda^{min}_{\Sigma_{(p,q)}})^2\label{n_3}
\end{align}
Here Eqs.~(\ref{n_1},\ref{n_3}) follow from $\lambda^{max}_X = \lambda^{min}_{X^{-1}}$ for invertible $X$. Eq.~(\ref{n_2}) follows from Eq.~(\ref{covar_varepsilon}) and using $\lambda^{min}_{X+Y}\geq \lambda^{min}_{X}$ for positive definite matrices $X, Y$.
\end{proof}
If noise across different buses is uncorrelated, then each block of $\Sigma_{(n_v,n_\theta)}$ in Eq.~(\ref{deviatesigma1}) will be a diagonal matrix. This is similar to the structure of $\Sigma_{(p,q)}$ in Eq.~(\ref{sigma_lc}) under Assumption $1$. In that setting, the following result follows immediately from Lemma \ref{noise_bound} through eigenvalue analysis of $2\times2$ matrices.
\begin{corollary} \label{col_noise}
Consider $\Delta\Sigma^{-1}_{(v,\theta)}=\Sigma^{-1}_{(\Tilde{v},\Tilde{\theta})}-\Sigma^{-1}_{(v,\theta)}$ where the noise $(n_v,n_\theta)$ is uncorrelated across buses and injections $(p,q)$ follow Assumption $1$. Let
\begin{align}
\footnotesize
\sigma^i_n &= \Sigma_{n_vn_v}(i,i)+\Sigma_{n_\theta n_\theta}(i,i)\label{col_n_sigma}\\
 ~~~&+\sqrt{(\Sigma_{n_vn_v}(i,i)-\Sigma_{n_\theta n_\theta}(i,i))^2+4\Sigma^2_{n_vn_\theta}(i,i)}\\
 \sigma^i_{(p,q)} &=\Sigma_{pp}(i,i)+\Sigma_{qq}(i,i)\nonumber\\
 ~~~&-\sqrt{(\Sigma_{pp}(i,i)-\Sigma_{qq}(i,i))^2+4\Sigma^2_{pq}(i,i)}. \label{col_p_sigma}
\end{align}
Then $\max_{i,j} |\Delta\Sigma^{-1}_{(v,\theta)}(i,j)| \leq \frac{2(\lambda^{max}_{H^2_{(g,\beta)}})^2 \max_i \sigma^i_n}{\min_i (\sigma^i_{(p,q)})^2}.$ \\
If $\Sigma_{pq}=0$ and $\Sigma_{n_vn_\theta} = 0$, then \\$\max_{i,j} |\Delta\Sigma^{-1}_{(v,\theta)}(i,j)| \leq \frac{(\lambda^{max}_{H^2_{(g,\beta)}})^2 \max_i \max(\Sigma_{n_vn_v}(i,i), \Sigma_{n_\theta n_\theta}(i,i))}{\min_i \min(\Sigma^2_{pp}(i,i),\Sigma^2_{qq}(i,i))}.$
\end{corollary}
Corollary \ref{col_noise} intuitively states that the bound on the largest deviation in inverse covariance of voltages increases with the noise covariance, but decreases with the injection covariance. Our learning algorithms depend on relative values in the estimated inverse covariance matrix of voltages. We use Lemma \ref{noise_bound} and Corollary \ref{col_noise} to determine maximum values in $\Delta\Sigma^{-1}_{(v,\theta)}$ and noise covariance $\Sigma_{(n_v,n_\theta)}$ and selection of thresholds to ensure correct learning under noise.
\begin{theorem}\label{noise_est}
In grid $\mathcal{G}$, consider estimated inverse covariance matrix of voltages $\Sigma^{-1}_{(\Tilde{v},\Tilde{\theta})}$, under uncorrelated zero-mean noise of covariance $\Sigma_{(n_v,n_\theta)}$. For the noiseless inverse covariance matrix of voltages $\Sigma^{-1}_{(v,\theta)}$, given by Lemma \ref{inverse_LC_PF}, define $\gamma_1,\gamma_2$ as:
\begin{align}
\gamma_1 &\coloneqq \min_{i \neq j} \{|J_{vv}(i,j)|, J_{vv}(i,j)\neq 0\}\label{gamma1}\\
\gamma_2 &\coloneqq \min_{i \neq j} \{|J_{vv}(i,j)+J_{\theta\theta}(i,j)|, J_{vv}(i,j)+ J_{\theta\theta}(i,j)<0\}\label{gamma2}
\end{align}
1. If $\max_{i,j} |\Delta\Sigma^{-1}_{(v,\theta)}(i,j)| < \gamma_1/2$, then at large sample sizes, Algorithm $1$ with threshold $\tau_1 =\gamma_1/2$ outputs the correct topology, with noise.\\
2. If $\max_{i,j} |\Delta\Sigma^{-1}_{(v,\theta)}(i,j)| < \gamma_2/4$, then Algorithm $2$ with threshold $\tau_2 =\gamma_2/2$ will output the correct topology, with noise, in the large sample limit.\\
3. Let noise be uncorrelated across buses with $\sigma^i_n,\sigma^i_{(p,q)}$ defined in Eqs.~(\ref{col_n_sigma},\ref{col_p_sigma}) respectively. Algorithms $1$ and $2$ will output correctly, if $\max_i\sigma^i_n< \frac{\min(2\gamma_1,\gamma_2)\min_i (\sigma^i_{(p,q)})^2}{8(\lambda^{max}_{H^2_{(g,\beta)}})^2}$.
\end{theorem}
\begin{proof}
1. $J_{vv}(i,j) \neq 0$ for edge $(ij)$ in `hybrid' graphical model $\mathcal{GM}_{hd}$ of grid $\mathcal{G}$, and equals $0$ otherwise. Let $\max_{i,j} |\Delta\Sigma^{-1}_{(v,\theta)}(i,j)| < \gamma_1/2$. Using Eqs.~(\ref{deviatesigma1},\ref{deviateinverse}), we have
\begin{align}
 &|J_{\Tilde{v}\Tilde{v}}(i,j)|\leq \max_{i,j} |\Delta\Sigma^{-1}_{(v,\theta)}(i,j)|< \gamma_1/2 ~\forall (ij)\notin \mathcal{GM}_{hd}\nonumber\\
 &|J_{\Tilde{v}\Tilde{v}}(i,j)|\geq\gamma_1-\max_{i,j} |\Delta\Sigma^{-1}_{(v,\theta)}(i,j)|> \gamma_1/2 ~\forall (ij) \in \mathcal{GM}_{hd}.\nonumber
\end{align}
Thus, Algorithm $1$ with threshold $\tau_1 =\gamma_1/2$, under noise, correctly identifies $\mathcal{GM}_{hd}$ and consequently grid $\mathcal{G}$, in the large sample limit.\\
2. $J_{vv}(i,j) +J_{\theta\theta}(i,j) < 0$ for edge $(ij)$ in grid $\mathcal{G}$, and takes a non-negative value otherwise. Using $\max_{i,j} |\Delta\Sigma^{-1}_{(v,\theta)}(i,j)| < \gamma_2/4$, with Eqs.~(\ref{deviatesigma1},\ref{deviateinverse}), we have
\begin{align}
 &J_{\Tilde{v}\Tilde{v}}(i,j)+J_{\Tilde{\theta}\Tilde{\theta}}(i,j)< -\gamma_2+2\gamma_2/4 = -\gamma_2/2 ~~\forall (ij)\in \mathcal{G}\nonumber\\
 &J_{\Tilde{v}\Tilde{v}}(i,j)+J_{\Tilde{\theta}\Tilde{\theta}}(i,j)> 0-2\gamma_2/4 = -\gamma_2/2~~\text{otherwise}.\nonumber
\end{align}
Thus, in the large sample limit, threshold $\tau_2 =\gamma_2/2$ in Algorithm $2$ enables correct estimation of grid $\mathcal{G}$, under noise.
3. This follows directly from statements 1, 2 and Corollary \ref{col_noise}.
\end{proof}}
{\subsection{Effect of Injection Correlation}\label{sec:corr}
The theorems guaranteeing the asymptotic correctness of our algorithms rely on nodal injection fluctuations being uncorrelated. We now analyze their performance under correlated injections. Consider the setting where the correlated inverse covariance matrix $\Sigma^{-1}_{(p,q)^c} = \Sigma^{-1}_{(p,q)} + \Delta\Sigma^{-1}_{(p,q)}$. Here $\Sigma^{-1}_{(p,q)}$ is the diagonal uncorrelated inverse covariance matrix. $\Delta\Sigma^{-1}_{(p,q)}$ is the difference matrix that arises from dependence among injections fluctuations. Following Eq.~\ref{covar_varepsilon}, the inverse covariance of voltages under $\Sigma^{-1}_{(p,q)^c}$ is given by
\begin{align}
 \Sigma^{-1}_{(v,\theta)^c} = \Sigma^{-1}_{(v,\theta)} + H_{(g,\beta)}\Delta\Sigma^{-1}_{(p,q)}H_{(g,\beta)}.\label{corr_v}
\end{align}
Here $\Sigma^{-1}_{(v,\theta)}$ is the inverse covariance of voltages under uncorrelated $\Sigma^{-1}_{(p,q)}$. Its algebraic form is given by Lemma \ref{inverse_LC_PF}, with related terms $\gamma_1,\gamma_2$ defined in Eqs.~\ref{gamma1},\ref{gamma2}. The following result holds.
\begin{lemma}\label{corr_est}
In grid $\mathcal{G}$, consider inverse covariance matrix of voltages $\Sigma^{-1}_{(v,\theta)^c}$, under correlated injection covariance $\Sigma^{-1}_{(p,q)^c}$. \\
1. If $\max_{i,j} |H_{(g,\beta)}\Delta\Sigma^{-1}_{(p,q)}H_{(g,\beta)}| < \gamma_1/2$, then at large sample sizes, Algorithm $1$ with threshold $\tau_1 =\gamma_1/2$ outputs the correct topology.\\
2. If $\max_{i,j} |H_{(g,\beta)}\Delta\Sigma^{-1}_{(p,q)}H_{(g,\beta)}| < \gamma_2/4$, then at large sample sizes, Algorithm $2$ with threshold $\tau_2 =\gamma_2/2$ will output the correct topology.
\end{lemma}
The proof uses Eq.~\ref{corr_v} to bound the maximum deviation in entries in $\Sigma^{-1}_{(v,\theta)^c}$, due to injection correlation. The remaining analysis for correctness is similar to the proof of Theorem \ref{noise_est}, and is omitted for space reasons.
Observe that the maximum injection correlation for which the algorithms are consistent, depends on the system parameters (topology and impedances). One can also consider a different model where the correlated injection covariance (not its inverse) is given as $\Sigma_{(p,q)^c}+ \Delta\Sigma_{(p,q)}$. A perturbation analysis as described in Eq.~\ref{deviateinverse} can be used in that setting to arrive at corresponding bounds.} In the next section, we discuss the accuracy of our methods and its dependence on number of voltage samples considered.\\
{\textbf{Note:} The thresholds determined in Theorem \ref{noise_est} and Lemma \ref{corr_est} operate in the large sample limit. For finite number of samples, thresholds need to be tuned using historical data to ensure correct estimation. Bounds for correct change detection, and selection of threshold $\tau_3$ in Algorithm $3$ can be analyzed similarly. We omit it due to space limitations.}
\section{Numerical Simulations}
\label{sec:simulations}
While the theoretical analysis for our designed algorithms are based on linearized power flow models, we demonstrate the practicality of our algorithms using non-linear AC samples for different test cases, {in the presence of noise.} Nodal injection fluctuations centered around the mean injection are modelled by uncorrelated zero-mean Gaussian random variables with covariance taken as $1e^{-2}$, relative to the base injection. The injection samples are sent to non-linear Matpower \cite{matpower} solver to generate i.i.d. samples of nodal voltages. {We also corrupt the nodal voltage measurements with Gaussian noise of different variances (measured relative to the nodal variance).} These measurements are used to estimate the inverse covariance matrix $\Sigma^{-1}_{(\Tilde{v},\Tilde{\theta})}$ of voltages through Graphical Lasso as described in previous sections. The estimated $\Sigma^{-1}_{(\Tilde{v},\Tilde{\theta})}$ is then input to Algorithms $1, 2$ and $3$. The thresholds $\tau_1,\tau_2$ and $\tau_3$ for the algorithms are empirical tuned to values which give minimal estimation errors in the large sample limit. They are then fixed while the measurement sample sizes are varied.
\subsection{Topology Estimation}
We first discuss results for Algorithm $1$ (neighborhood search) and Algorithm $2$ (sign rule) in learning the operational network. We consider a modified case with $56$ nodes \cite{bolognani2013identification} derived from the IEEE $123$ test feeder \cite{kersting2001radial}. We consider loopy extensions of this system with differing minimum cycle lengths.
\begin{figure}[htb]
\centering\hfill
\subfigure[]{\includegraphics[width=0.07\textwidth]{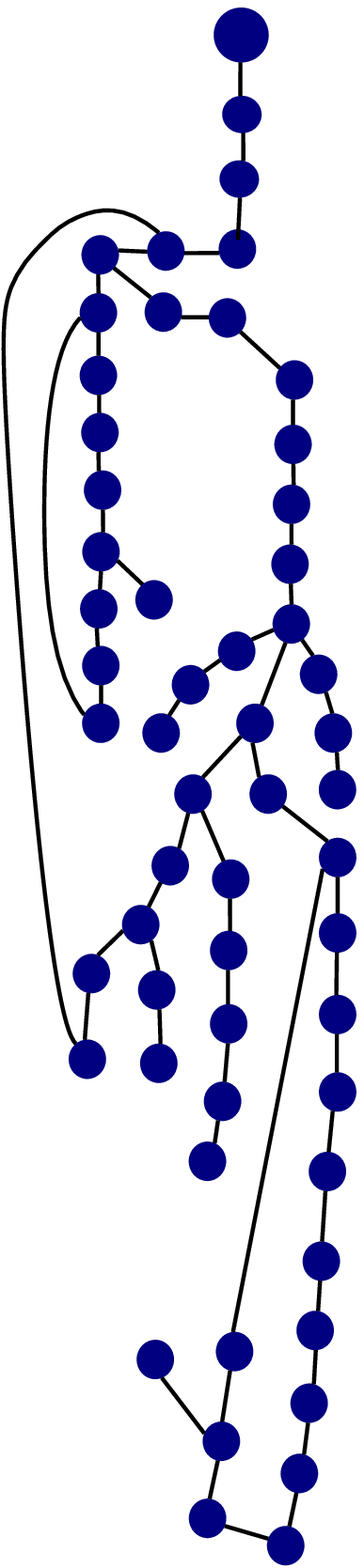}\label{fig:case56_7}}~~
\subfigure[]{\includegraphics[width=0.42\textwidth]{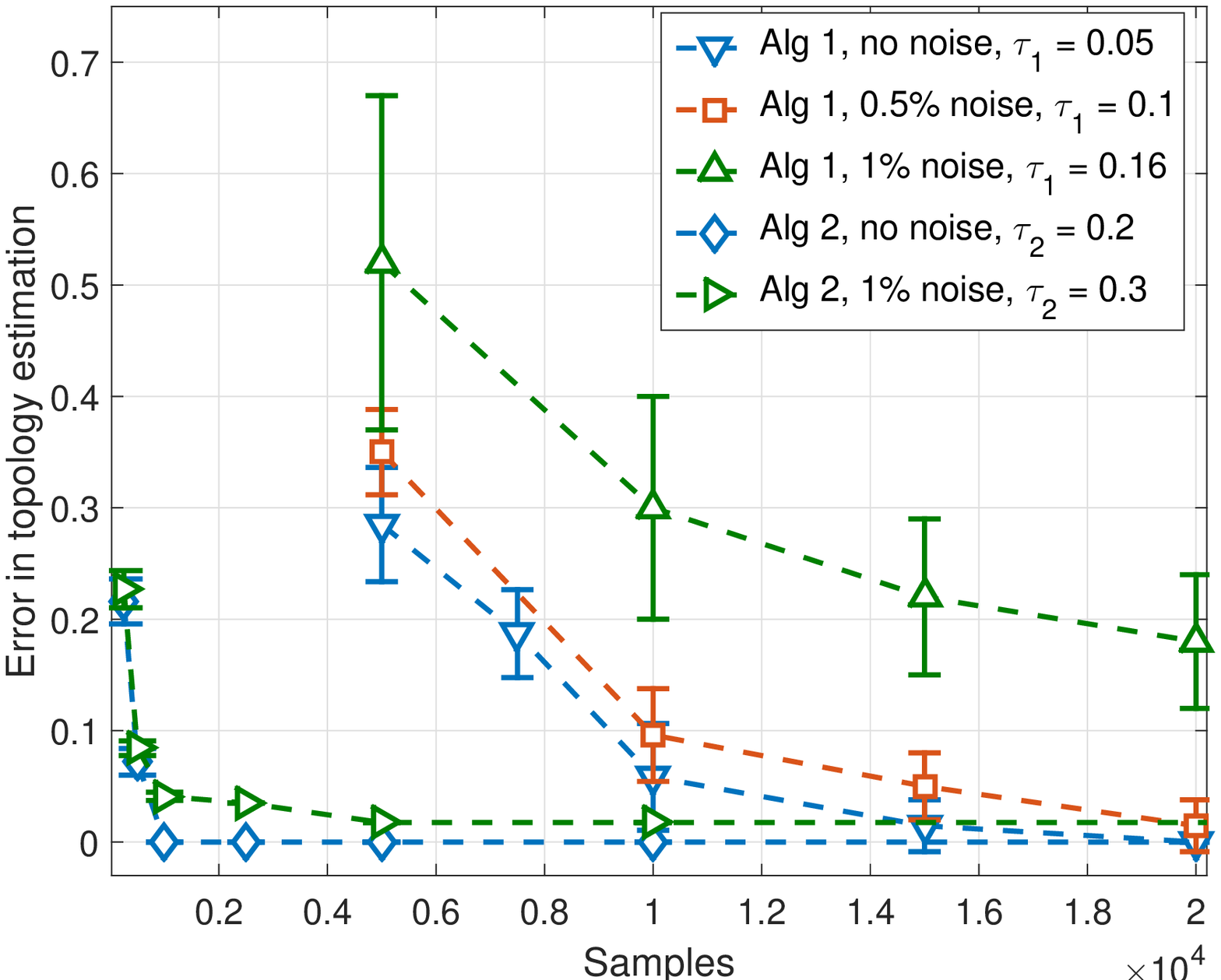}\label{fig:errors_56_7}}
\vspace{-.25cm}
\caption{(a) Modified $56$ node test case \cite{bolognani2013identification} with cycle length $7$ (b) Average errors in Algorithm $1,2$ for topology learning with number of {noisy} voltage samples over $10$ independent runs.}\label{fig:57_7}
\vspace{-3mm}
\end{figure}
First consider the system with $3$ loops with minimum cycle length of $7$ as shown in Fig.~\ref{fig:case56_7}. We apply Algorithm 1 and Algorithm 2 for topology inference, and demonstrate the effect of number of samples on estimation accuracy in 10 independent realizations in Fig.~\ref{fig:errors_56_7}. Error in each run is computed as the ratio of the sum of number of false links identified and true links missed, to the total number of links in the underlying network. Observe that both algorithms give zero errors in the large sample limit, {though the performance deteriorates with higher noise.} This is consistent with Theorems \ref{condind_loopy},~\ref{condind_leaf},~\ref{threshld_loopy} which prove exact recovery of either algorithm for cycle lengths greater than $6$. However, in the limited sample regime, the sign rule based algorithm clearly outperforms the neighborhood search approach {Among others, this is due to the fact that inaccurate estimation (presence or absence) of two hop neighbors affects the discovery of true edges in Algorithm $1$. On the other hand, Algorithm $2$ uses a threshold to determine the true grid edges directly and can hence sustain greater inaccuracy in entries corresponding to two hop neighbors.}

Next, we consider the case in Fig. \ref{fig:case56_4} with cycle length of $4$. The errors in topology estimation for Algorithms $1,2$ are shown in Fig.~\ref{fig:errors_56_4}. {While Algorithm $2$ (sign rules) gives exact recovery at high samples for noiseless samples, Algorithm $1$ fails to do so as the minimum cycle length here is lower than $7$. Note that Algorithm $1$ is able to give exact recovery for $1\%$ noise case, but saturates to a non-zero error for the $2\%$ noise case. The performance under noise depends on the values in the system matrices and injection covariances as presented in the previous section.}
\begin{figure}[htb]
\centering\hfill
\subfigure[]{\includegraphics[width=0.07\textwidth]{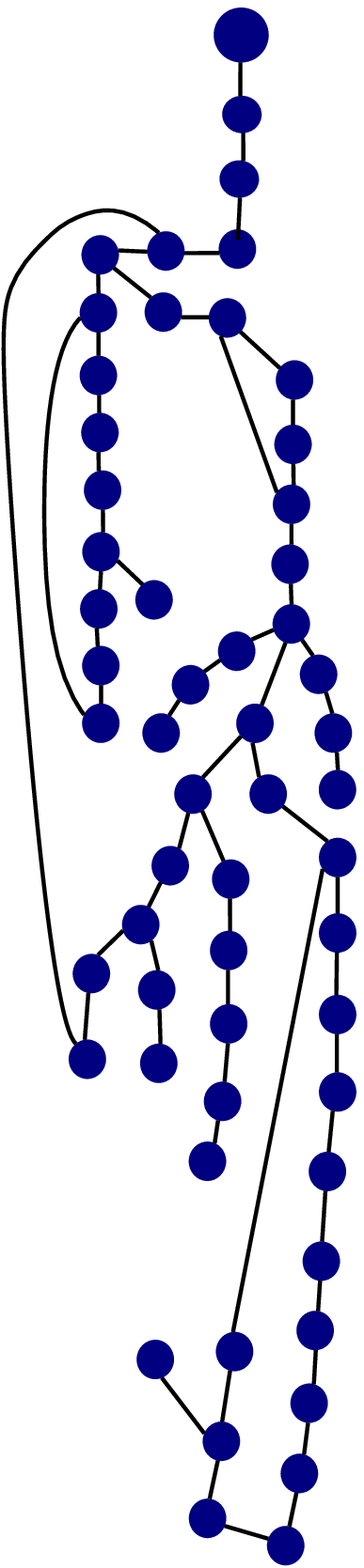}\label{fig:case56_4}}~~
\subfigure[]{\includegraphics[width=0.42\textwidth]{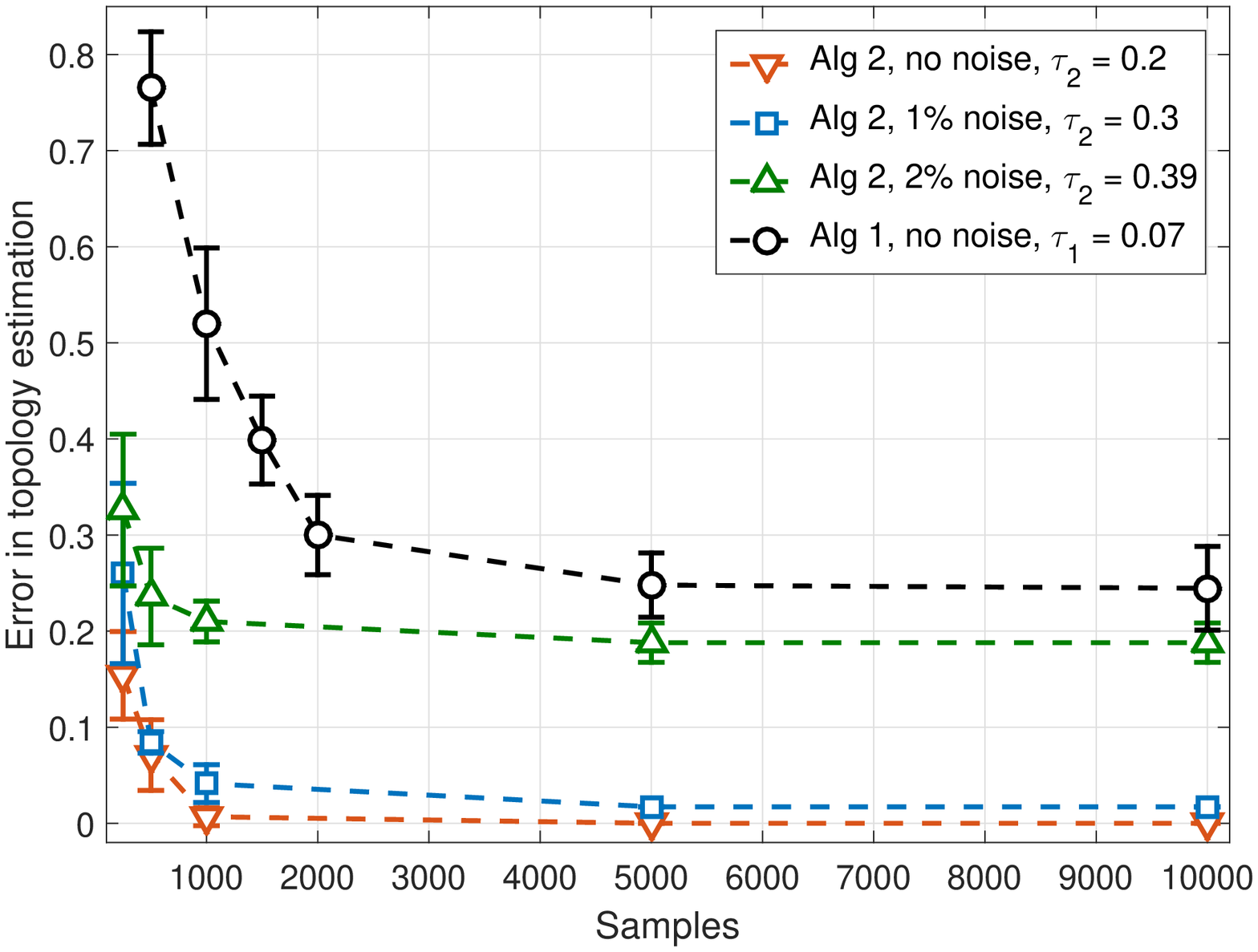}\label{fig:errors_56_4}}
\vspace{-.25cm}
\caption{(a) Modified $56$ node test case \cite{bolognani2013identification} with cycle length $4$ (b) Average errors in Algorithm $1,2$ for topology learning with number of {noisy} voltage samples over $10$ independent runs.}\label{fig:57_4}
\vspace{-3mm}
\end{figure}
Finally, we consider a modified $33$ node system \cite{matpower} which is made loopy with $3$ node cycles as shown in Fig.~\ref{fig:case33_3}. This case violates the minimum cycle length necessary for exact theoretical recovery in both algorithms. The accuracy at high samples indeed does not decay to zero, even in the noiseless case, for both of them, as shown in Fig.~\ref{fig:errors_33_3} for 10 independent realizations. These simulation results thus validate the theoretical consistency and restrictions for exact topology learning for loopy power grids.
\begin{figure}[htb]
\centering\hfill
\subfigure[]{\includegraphics[width=0.07\textwidth]{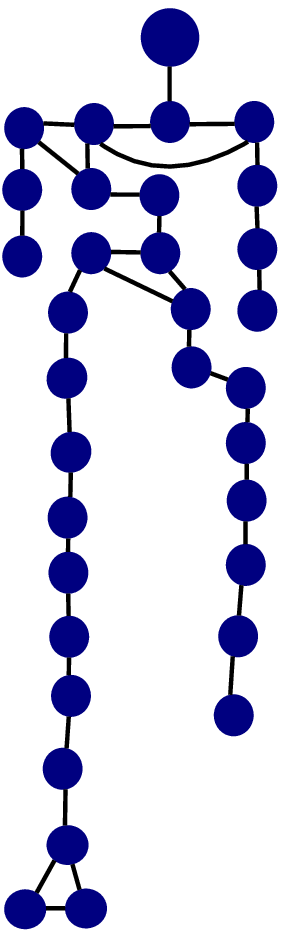}\label{fig:case33_3}}~~
\subfigure[]{\includegraphics[width=0.42\textwidth]{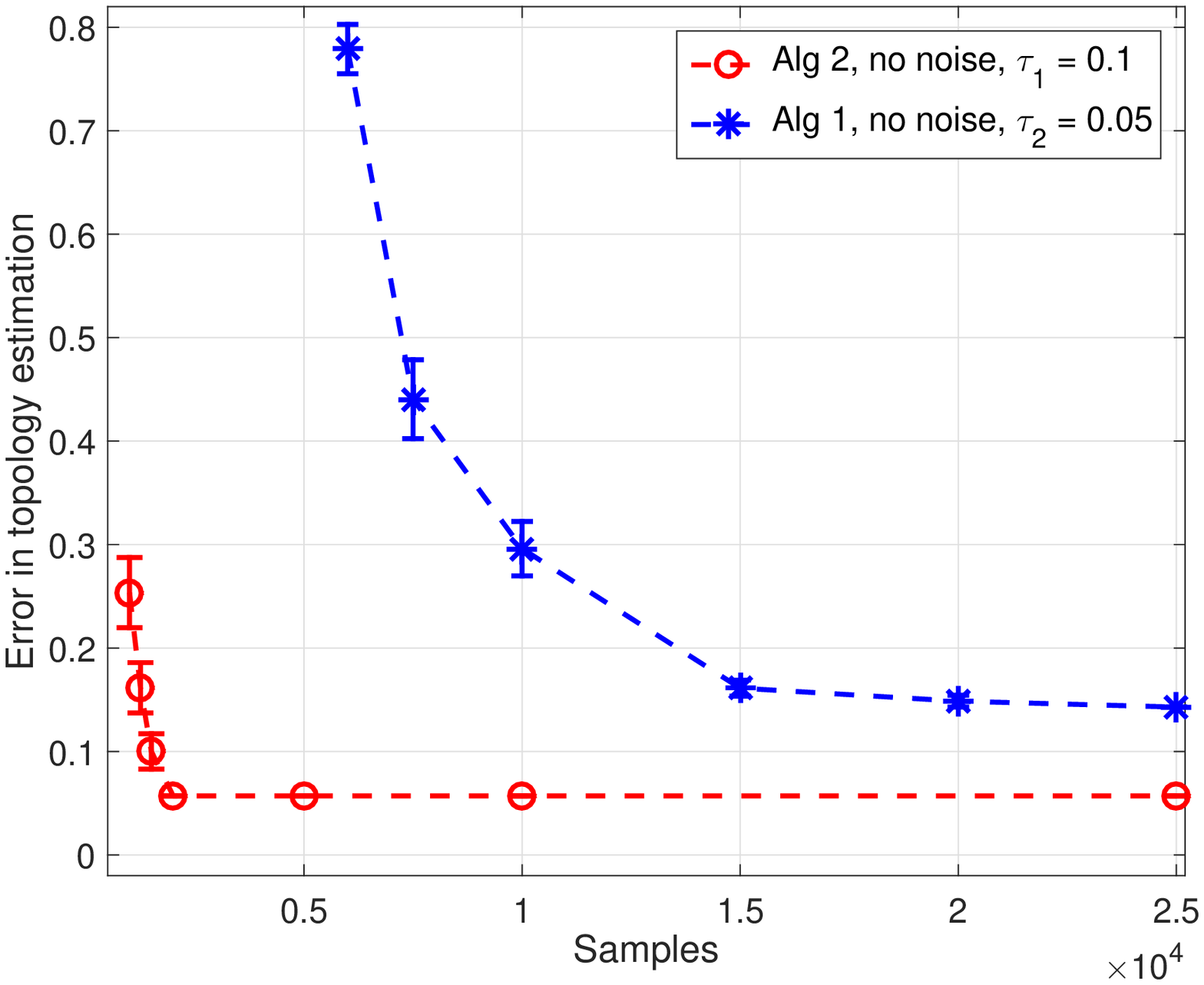}\label{fig:errors_33_3}}
\vspace{-.25cm}
\caption{(a) modified $33$ bus networks with cycle length $3$ (b) Average errors in Algorithm $1,2$ for topology learning with number of {noiseless} voltage samples over $10$ independent runs.}\label{fig:33_3}
\vspace{-3mm}
\end{figure}

{\textbf{Effect of Threshold Selection and Injection Correlation:} The results for Algorithms $1,2$ include performance under optimized thresholds, with voltage samples generated under uncorrelated injections (Assumption $1$). We now analyze the sensitivity of the performance to both threshold selection and injection correlation. We consider the grid in Fig.~\ref{fig:case56_7} with cycle length $7$, for which Algorithms $1$ and $2$ theoretically output the correct topology in the asymptotic limit. To analyze the sensitivity to thresholds, we fix the number of measurements to large values under which recovery is exact (see Fig.~\ref{fig:errors_56_7}). We then change thresholds $\tau_1, \tau_2$ in Algorithms $1$ and $2$ respectively around their optimal values, and report the estimation errors over $10$ independent simulations. It is clear from Fig.~\ref{fig:thres_effect} that both algorithms give exact reconstruction for varying the thresholds by $20\%$ or more, around the selected values, though the errors increase if the variation is much larger. Note that Algorithm $2$ is more robust to threshold variation for the test-case considered.}
\begin{figure}[htb]
\centering
\subfigure[]{\includegraphics[width=0.42\textwidth,height=0.23\textwidth]{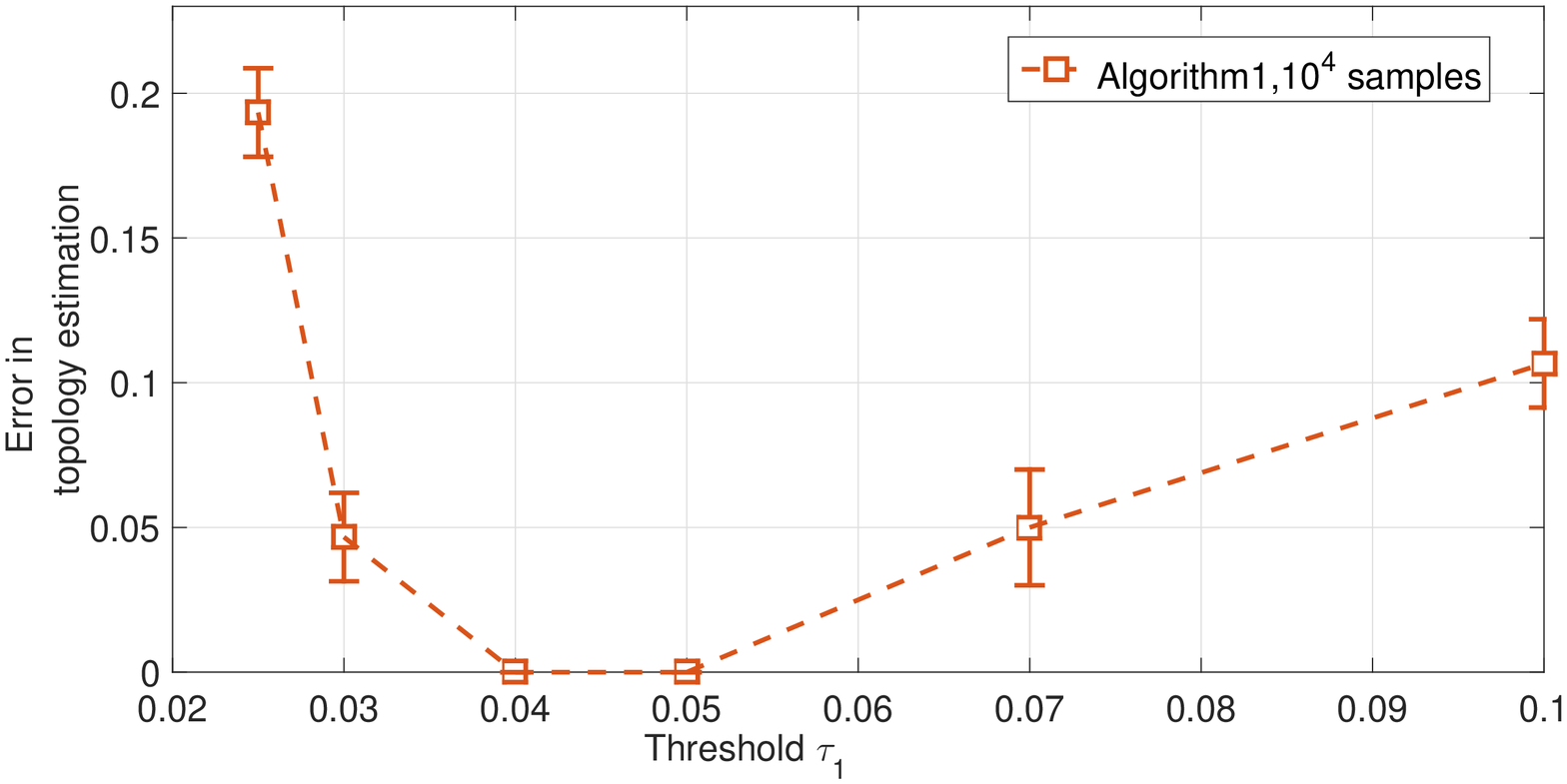}\label{fig:alg1_thres}}\hfill
\vspace{-.25cm}\subfigure[]{\includegraphics[width=0.42\textwidth,height=0.23\textwidth]{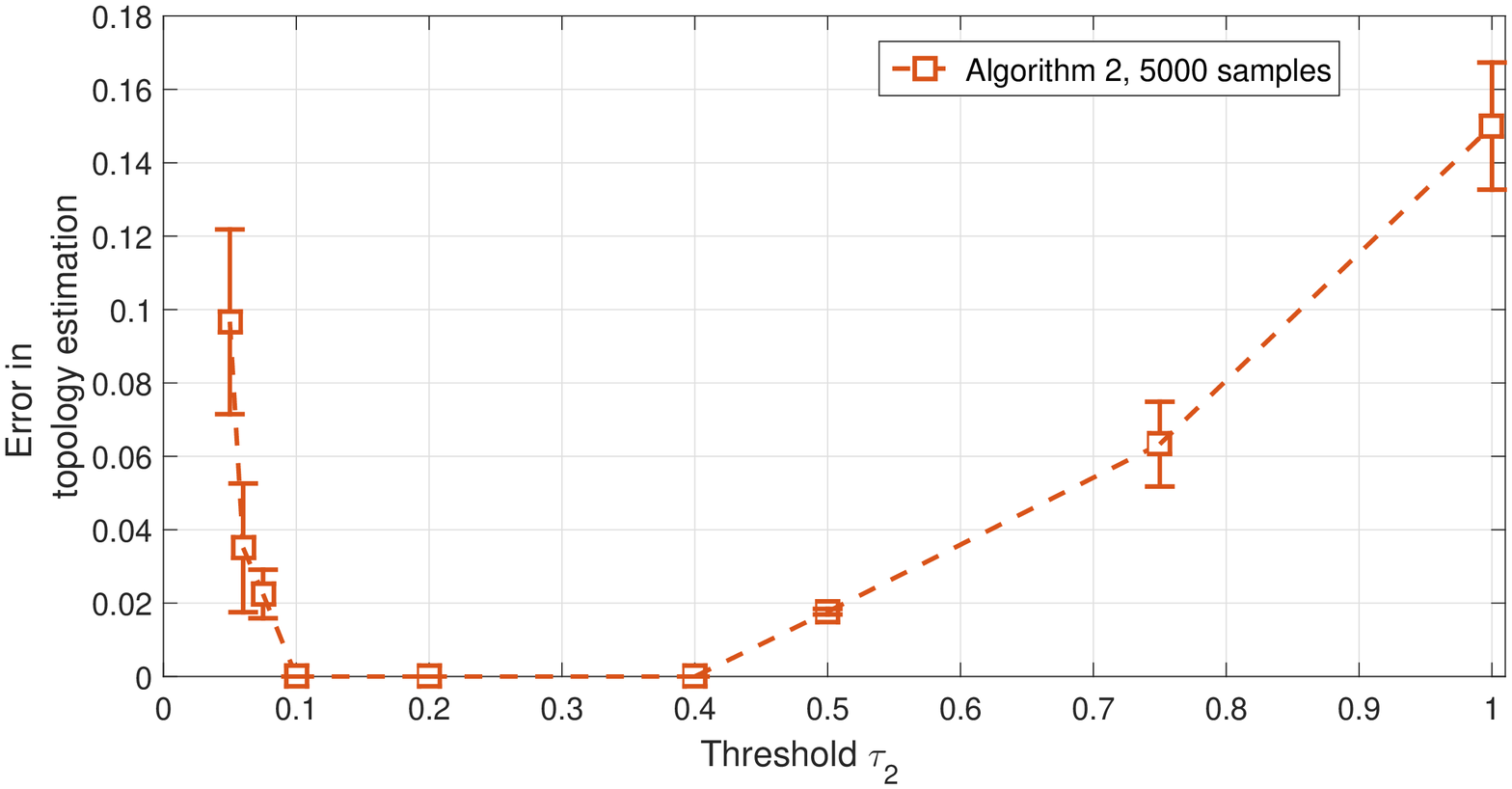}\label{fig:alg2_thres}}
\caption{Sensitivity to threshold selection for Algorithms (a) $1$ and (b) $2$ for grid in Fig.~\ref{fig:case56_7} for large sample sizes, over $10$ independent runs.}
\label{fig:thres_effect}
\vspace{-3mm}
\end{figure}

{Next we consider the performance of the algorithms under correlated injections for the test case in Fig.~\ref{fig:case56_7}. To the inverse of the injection covariance matrix, we add off-diagonal terms (see Section \ref{sec:corr}) of value $\epsilon$ relative to the diagonal entries. We then use it to generate injection and corresponding Matpower voltage samples. Fig.~\ref{fig:corr_effect} includes the performance of Algorithms $1,2$ for different samples sizes, under increasing values of the off-diagonal terms. Observe that while the performance deteriorates with the increased injection correlation, for $10\%$ relative value of off-diagonal terms, the relative errors are within $.1$ for both algorithms.}
\begin{figure}[htb]
\centering
\subfigure[]{\includegraphics[width=0.44\textwidth,height=0.23\textwidth]{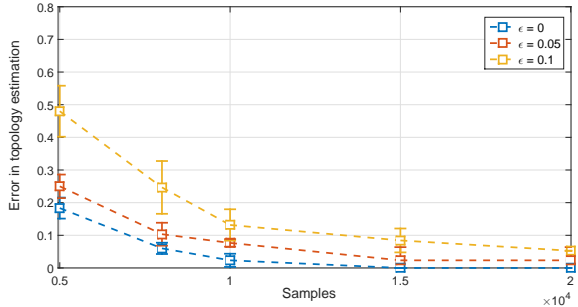}\label{fig:alg1_corr}}\hfill
\vspace{-.25cm}\subfigure[]{\includegraphics[width=0.44\textwidth,height=0.23\textwidth]{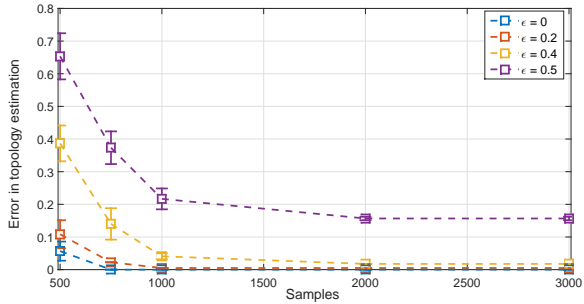}\label{fig:alg2_corr1}}
\caption{Sensitivity to injection correlation for Algorithms $1$ and $2$ for grid in Fig.~\ref{fig:case56_7} for different sample sizes, over $10$ independent runs. $\epsilon$ represents the relative value of off-diagonal elements in inverse injection covariance matrix to the diagonal covariance terms.}
\label{fig:corr_effect}
\vspace{-3mm}
\end{figure}

\subsection{Change Detection}
\begin{figure}[htb]
\centering\hfill
\subfigure[]{\includegraphics[width=0.09\textwidth]{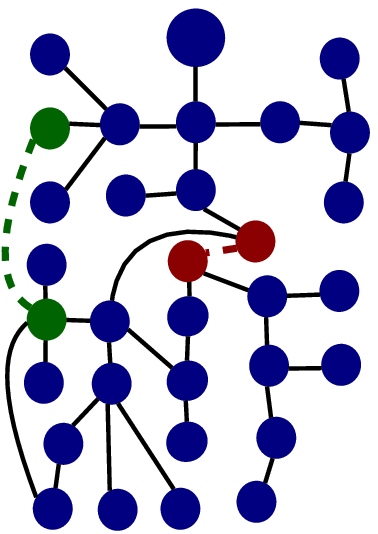}\label{fig:case33_remove}}~~
\subfigure[]{\includegraphics[width=0.38\textwidth]{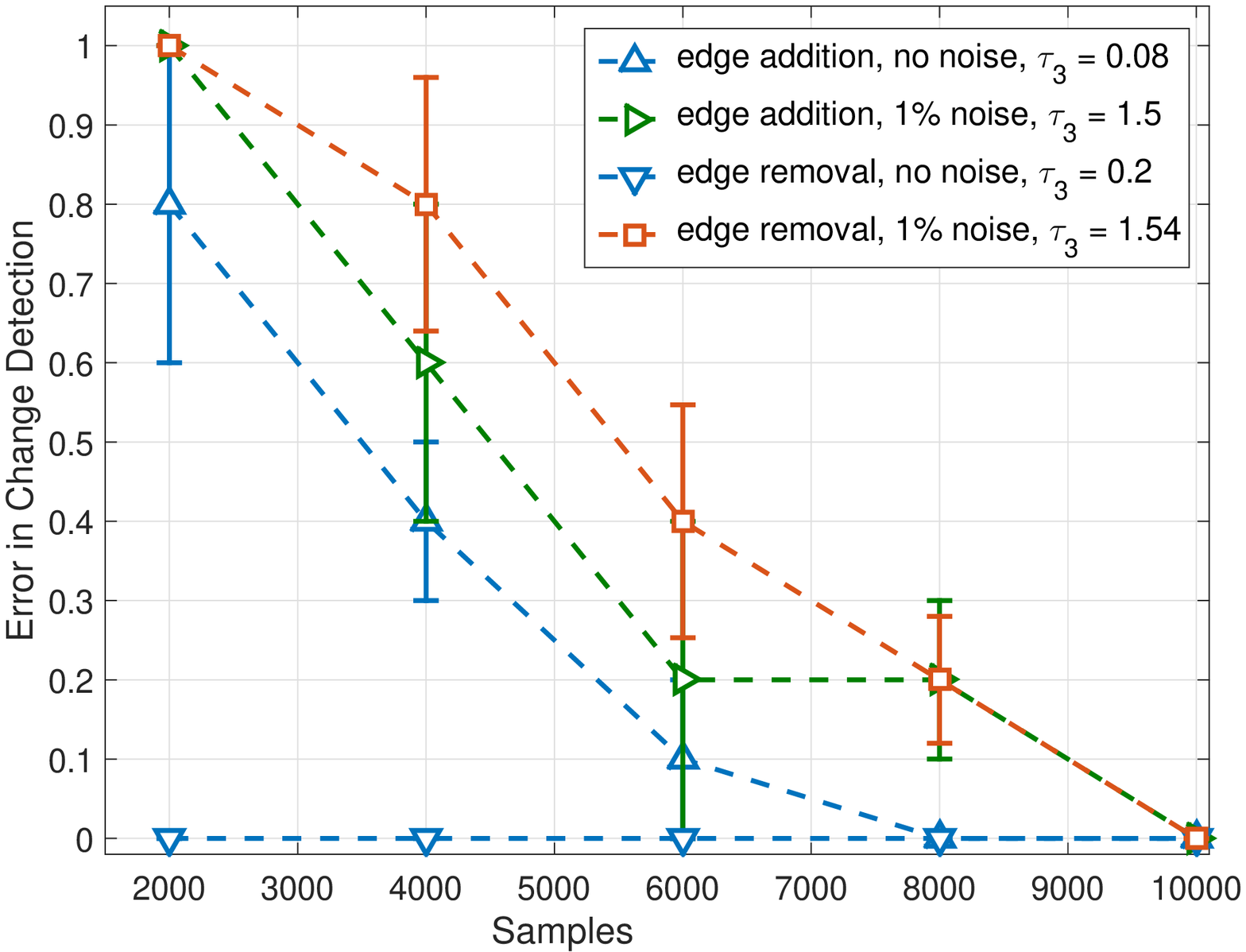}\label{fig:errors_removal}}
\vspace{-.25cm}
\caption{(a) Modified $33$ bus network with edge $(7,20)$ added (colored green) and edge $(5,25)$ removed (colored red) (b) Accuracy of Algorithm $3$ in identifying line addition and line removal with number of voltage measurement samples{, with noise.}}
\label{fig:errors_change}
\vspace{-3mm}
\end{figure}
We now consider Algorithm $3$ for topology change detection and present simulation results for a modified 33 bus distribution system \cite{matpower} with loops, shown in Fig.~\ref{fig:case33_remove}. We consider the cases where an edge between buses $8$ and $21$ is added, and the edge $(6,26)$ is removed. As before, voltage samples generated by Matpower {and corrupted with noise}, are used to estimate the voltage inverse covariance for change detection for edge addition and removal. We demonstrate in Figure~\ref{fig:errors_removal}, the accuracy of detecting the correct edges for various sample sizes from 10 realizations. In each simulation run, an error of $1$ is recorded if the line addition or removal goes undetected or a wrong line is identified, while error equals $0$ if the correct line is identified. Clearly in the large sample limit, the change is detected for every realization, though the performance under noise at low samples is worse than the performance without noise.

\section{Conclusion}
\label{sec:conclusion}
In this paper, we discuss theoretical aspects of topology learning and change detection for general distribution grids that may have cycles. We present two learning algorithms based on nodal voltage graphical models that are able to estimate the grid topology under non-restrictive topological conditions. Crucially our theoretical methodology does not depend on the knowledge of line impedances and injections statistics, and further does not assume approximate probabilistic models or parameter restrictions. We extend the analysis to the problem to identification in the presence of noise, and correlated injections. Finally we present a topology change detection algorithms for grids using graphical models. Simulation results on non-linear AC power flow samples demonstrate the applicability of this work to realistic grid samples.

This work, while generalizing several approaches, opens multiple directions of future work. In particular, the work can be extended to the case of three phase power flow models of the type discussed in \cite{dekathreephase,lowlinear}. We are also interested in exploring topology change detection to the case of multiple line failure estimation, and with missing nodal measurements.
\section*{Acknowledgement}
The authors D. Deka, and M. Chertkov acknowledge the support from the Department of Energy through the Grid Modernization Lab Consortium, and the Center for Non Linear Studies (CNLS) at Los Alamos National Laboratory. The authors S. Talukdar, and M. V. Salapaka acknowledge the support of ARPA-E under projects titled `A Robust Distributed Framework for Flexible Power Grids’ via grant no. DE-AR000071, and `Rapidly Viable Sustained Grid' via grant no. DE-AR0001016.
\bibliography{sigproc,FIDVR,SmartGrid,voltage,trees}

\begin{thebibliography}{10}
\providecommand{\url}[1]{#1}
\csname url@samestyle\endcsname
\providecommand{\newblock}{\relax}
\providecommand{\bibinfo}[2]{#2}
\providecommand{\BIBentrySTDinterwordspacing}{\spaceskip=0pt\relax}
\providecommand{\BIBentryALTinterwordstretchfactor}{4}
\providecommand{\BIBentryALTinterwordspacing}{\spaceskip=\fontdimen2\font plus
\BIBentryALTinterwordstretchfactor\fontdimen3\font minus
  \fontdimen4\font\relax}
\providecommand{\BIBforeignlanguage}[2]{{%
\expandafter\ifx\csname l@#1\endcsname\relax
\typeout{** WARNING: IEEEtran.bst: No hyphenation pattern has been}%
\typeout{** loaded for the language `#1'. Using the pattern for}%
\typeout{** the default language instead.}%
\else
\language=\csname l@#1\endcsname
\fi
#2}}
\providecommand{\BIBdecl}{\relax}
\BIBdecl

\bibitem{hoffman2006practical}
R.~Hoffman, ``Practical state estimation for electric distribution networks,''
  in \emph{Power Systems Conference and Exposition, 2006. PSCE'06. 2006 IEEE
  PES}.\hskip 1em plus 0.5em minus 0.4em\relax IEEE, 2006, pp. 510--517.

\bibitem{NY}
C.~Rudin and et~al., ``Machine learning for the new york city power grid,''
  \emph{IEEE transactions on pattern analysis and machine intelligence},
  vol.~34, no.~2, 2011.

\bibitem{germany}
D.~Wolter, M.~Zdrallek, M.~St{\"o}tzel, C.~Schacherer, I.~Mladenovic, and
  M.~Biller, ``Impact of meshed grid topologies on distribution grid planning
  and operation,'' \emph{CIRED-Open Access Proceedings Journal}, pp.
  2338--2341, 2017.

\bibitem{taiwan}
J.-H. Teng, ``Unsymmetrical short-circuit fault analysis for weakly meshed
  distribution systems,'' \emph{IEEE Transactions on Power Systems}, vol.~25,
  no.~1, pp. 96--105, 2009.

\bibitem{distributiongridbook}
T.~Chee-Wooi and Y.~Tang, \emph{{Electric Power: Distribution Emergency
  Operation}}.\hskip 1em plus 0.5em minus 0.4em\relax CRC Press, 2018.

\bibitem{hassan2016topology}
H.~L. Hijazi and S.~Thi{\'e}baux, ``Optimal ac distribution systems
  reconfiguration,'' in \emph{Power Systems Computation Conference (PSCC),
  2014}.\hskip 1em plus 0.5em minus 0.4em\relax IEEE, 2014, pp. 1--7.

\bibitem{PMU}
A.~Phadke, ``Synchronized phasor measurements in power systems,'' \emph{IEEE
  Computer Applications in Power}, vol.~6, no.~2, pp. 10--15, 1993.

\bibitem{micropmu}
A.~von Meier, D.~Culler, A.~McEachern, and R.~Arghandeh, ``Micro-synchrophasors
  for distribution systems,'' in \emph{Innovative Smart Grid Technologies
  Conference (ISGT)}.\hskip 1em plus 0.5em minus 0.4em\relax IEEE, 2014.

\bibitem{FNET}
Z.~Zhong and et~al., ``Power system frequency monitoring network (fnet)
  implementation,'' \emph{Power Systems, IEEE Transactions on}, vol.~20, no.~4,
  2005.

\bibitem{dekatcns}
D.~Deka, S.~Backhaus, and M.~Chertkov, ``Structure learning in power
  distribution networks,'' \emph{IEEE Transactions on Control of Network
  Systems}, vol.~5, no.~3, pp. 1061--1074, 2018.

\bibitem{dekapscc}
D.~Deka, S.~Backhaus, and M.~Chertkov, ``Estimating distribution grid
  topologies: A graphical learning based approach,'' in \emph{Power Systems
  Computation Conference (PSCC), 2016}.\hskip 1em plus 0.5em minus 0.4em\relax
  IEEE, 2016, pp. 1--7.

\bibitem{arya}
V.~Arya, T.~Jayram, S.~Pal, and S.~Kalyanaraman, ``Inferring connectivity model
  from meter measurements in distribution networks,'' in \emph{the fourth
  international conference on Future energy systems}.\hskip 1em plus 0.5em
  minus 0.4em\relax ACM, 2013.

\bibitem{berkeley}
G.~Cavraro, R.~Arghandeh, A.~von Meier, and K.~Poolla, ``Data-driven approach
  for distribution network topology detection,'' \emph{arXiv preprint
  arXiv:1504.00724}, 2015.

\bibitem{cavraro2018graph}
G.~Cavraro and V.~Kekatos, ``Graph algorithms for topology identification using
  power grid probing,'' \emph{IEEE control systems letters}, vol.~2, no.~4,
  2018.

\bibitem{ramstanford}
R.~Sevlian and R.~Rajagopal, ``Feeder topology identification,'' \emph{arXiv
  preprint arXiv:1503.07224}, 2015.

\bibitem{dekasmartgridcomm}
D.~Deka, S.~Backhaus, and M.~Chertkov, ``Learning topology of distribution
  grids using only terminal node measurements,'' in \emph{IEEE Smartgridcomm},
  2016.

\bibitem{sejunpscc}
S.~Park, D.~Deka, and M.~Chertkov, ``Exact topology and parameter estimation in
  distribution grids with minimal observability,'' in \emph{Power Systems
  Computation Conference (PSCC), 2018}.\hskip 1em plus 0.5em minus 0.4em\relax
  IEEE, 2018.

\bibitem{sauravacc}
S.~Talukdar, D.~Deka, D.~Materassi, and M.~Salapaka, ``Exact topology
  reconstruction of radial dynamical systems with applications to distribution
  system of the power grid,'' in \emph{2017 American Control Conference
  (ACC)}.\hskip 1em plus 0.5em minus 0.4em\relax IEEE, 2017.

\bibitem{sauravacm}
S.~Talukdar, D.~Deka, B.~Lundstrom, M.~Chertkov, and M.~V. Salapaka, ``Learning
  exact topology of a loopy power grid from ambient dynamics,'' in
  \emph{Proceedings of the Eighth International Conference on Future Energy
  Systems}, ser. e-Energy '17, 2017, pp. 222--227.

\bibitem{bolognani2013identification}
S.~Bolognani, N.~Bof, D.~Michelotti, R.~Muraro, and L.~Schenato,
  ``Identification of power distribution network topology via voltage
  correlation analysis,'' in \emph{Decision and Control (CDC), 2013 IEEE 52nd
  Annual Conference on}.\hskip 1em plus 0.5em minus 0.4em\relax IEEE, 2013, pp.
  1659--1664.

\bibitem{dekathreephase}
D.~Deka, M.~Chertkov, and S.~Backhaus, ``Topology estimation using graphical
  models in multi-phase power distribution grids,'' \emph{IEEE Transactions on
  Power System}, 2019.

\bibitem{he2011dependency}
M.~He and J.~Zhang, ``A dependency graph approach for fault detection and
  localization towards secure smart grid,'' \emph{IEEE Transactions on Smart
  Grid}, vol.~2, no.~2, pp. 342--351, 2011.

\bibitem{weng2017distributed}
Y.~Weng, Y.~Liao, and R.~Rajagopal, ``Distributed energy resources topology
  identification via graphical modeling,'' \emph{IEEE Transactions on Power
  Systems}, vol.~32, no.~4, pp. 2682--2694, 2017.

\bibitem{liao2016urban}
Y.~Liao, Y.~Weng, G.~Liu, and R.~Rajagopal, ``Urban mv and lv distribution grid
  topology estimation via group lasso,'' \emph{IEEE Transactions on Power
  Systems}, vol.~34, no.~1, 2018.

\bibitem{wainwright2008graphical}
M.~J. Wainwright and M.~I. Jordan, ``Graphical models, exponential families,
  and variational inference,'' \emph{Foundations and Trends{\textregistered} in
  Machine Learning}, vol.~1, no. 1-2, pp. 1--305, 2008.

\bibitem{matpower}
\BIBentryALTinterwordspacing
``{IEEE 1547 Standard for Interconnecting Distributed Resources with Electric
  Power Systems}.'' [Online]. Available:
  \url{http://grouper.ieee.org/groups/scc21/1547/1547_index.html}
\BIBentrySTDinterwordspacing

\bibitem{dekairep}
D.~Deka, M.Chertkov, S.~Talukdar, and M.~V. Salapaka, ``Topology estimation in
  bulk power grids: Theoretical guarantees and limits,'' in
  \emph{Bulk Power System Dynamics and Control
  Symposium-IREP}, 2017.

\bibitem{89BWc}
M.~Baran and F.~Wu, ``Network reconfiguration in distribution systems for loss
  reduction and load balancing,'' \emph{Power Delivery, IEEE Transactions on},
  vol.~4, no.~2, pp. 1401--1407, Apr 1989.

\bibitem{bolognani2016existence}
S.~Bolognani and S.~Zampieri, ``On the existence and linear approximation of
  the power flow solution in power distribution networks,'' \emph{Power
  Systems, IEEE Transactions on}, vol.~31, no.~1, pp. 163--172, 2016.

\bibitem{89BWb}
M.~Baran and F.~Wu, ``Optimal capacitor placement on radial distribution
  systems,'' \emph{Power Delivery, IEEE Transactions on}, vol.~4, no.~1, pp.
  725--734, Jan 1989.

\bibitem{89BWa}
M.~Baran and F.~Wu, ``Optimal sizing of capacitors placed on a radial
  distribution system,'' \emph{Power Delivery, IEEE Transactions on}, vol.~4,
  no.~1, pp. 735--743, Jan 1989.

\bibitem{abur2004power}
A.~Abur and A.~G. Exposito, \emph{Power system state estimation: theory and
  implementation}.\hskip 1em plus 0.5em minus 0.4em\relax CRC Press, 2004.

\bibitem{cavraro2017voltage}
G.~Cavraro, V.~Kekatos, and S.~Veeramachaneni, ``Voltage analytics for power
  distribution network topology verification,'' \emph{IEEE Transactions on
  Smart Grid}, vol.~10, no.~1, pp. 1058--1067, 2017.

\bibitem{yury}
Y.~Dvorkin, M.~Lubin, S.~Backhaus, and M.~Chertkov, ``Uncertainty sets for wind
  power generation,'' \emph{IEEE Transactions on Power Systems}, vol.~31,
  no.~4, pp. 3326--3327, 2016.

\bibitem{bienstock2014chance}
D.~Bienstock, M.~Chertkov, and S.~Harnett, ``Chance-constrained optimal power
  flow: Risk-aware network control under uncertainty,'' \emph{Siam Review},
  vol.~56, no.~3, pp. 461--495, 2014.

\bibitem{gubner2006probability}
J.~A. Gubner, \emph{Probability and random processes for electrical and
  computer engineers}.\hskip 1em plus 0.5em minus 0.4em\relax Cambridge
  University Press, 2006.

\bibitem{yuan2007model}
M.~Yuan and Y.~Lin, ``Model selection and estimation in the gaussian graphical
  model,'' \emph{Biometrika}, vol.~94, no.~1, pp. 19--35, 2007.

\bibitem{tibshirani}
J.~Friedman, T.~Hastie, and R.~Tibshirani, ``Sparse inverse covariance
  estimation with the graphical lasso,'' \emph{Biostatistics}, vol.~9, no.~3,
  pp. 432--441, 2008.

\bibitem{ami_noise}
E.-A.-U. An, ``Smart meters and smart meter systems: A metering industry
  perspective,'' 2011.

\bibitem{pmu_noise}
M.~Brown, M.~Biswal, S.~Brahma, S.~J. Ranade, and H.~Cao, ``Characterizing and
  quantifying noise in pmu data,'' in \emph{2016 IEEE Power and Energy Society
  General Meeting (PESGM)}.\hskip 1em plus 0.5em minus 0.4em\relax IEEE, 2016,
  pp. 1--5.

\bibitem{hager1989updating}
W.~W. Hager, ``Updating the inverse of a matrix,'' \emph{SIAM review}, vol.~31,
  no.~2, pp. 221--239, 1989.

\bibitem{kersting2001radial}
W.~H. Kersting, ``Radial distribution test feeders,'' in \emph{Power
  Engineering Society Winter Meeting, 2001. IEEE}, vol.~2.\hskip 1em plus 0.5em
  minus 0.4em\relax IEEE, 2001, pp. 908--912.

\bibitem{lowlinear}
L.~Gan and S.~H. Low, ``Convex relaxations and linear approximation for optimal
  power flow in multiphase radial networks,'' in \emph{Power Systems
  Computation Conference (PSCC), 2014}.\hskip 1em plus 0.5em minus 0.4em\relax
  IEEE, 2014, pp. 1--9.

\end{thebibliography}
\appendix
\subsection{Proof of Theorem \ref{condind_loopy}}\label{proof_condind}
We first prove the \textbf{if} part. Let edge $(ij)$ exists in $\cal G$ between non-leaf nodes $i,j$. There exists nodes $k$ and $l$ {separated by three hops} that are neighbors of $i$ and $j$ respectively in $\cal G$ with path $k -i -j -l$ of length $3$. Any other path between $k, l$ in $\cal G$ must be longer than $3$ hops as minimum cycle length is $6$. Using Theorem \ref{structure_LC_PF}, edges $k-j,l-j, k-i,l-i \in \mathcal{GM}_{hd}$. Thus $k,l$ are two-hop neighbors in $\mathcal{GM}_{hd}$.

Next, we prove the \textbf{only if} part by contradiction. Suppose $(ij)$ is not an edge in $\mathcal{G}$.
Then, neighbors $i,j$ in $\mathcal{GM}_{hd}$ are two-hop neighbors in $\cal G$. As minimum cycle length is $6$, $i$ and $j$ can have exactly one common neighbor in $\cal G$, say node $c$. Now paths $k-i-l$ and $k-j-l$ in $\mathcal{GM}_{hd}$ exist, thus $k,l$ must be one or two-hop neighbors of both $i$ and $j$ in $\cal G$. As $k-l$ doesn't exist in $\mathcal{GM}_{hd}$, $k,l$ are separated by more than two hops in $\mathcal{G}$ and hence not both connected to $c$. Hence, cycles $i-r_1-(k/l)-r_2-j-c-i$ or $i-r_1-k-r_2-j-(l=c)-i$ must exist in $\mathcal{G}$ {with distinct nodes $r_1, r_2$ in the graph}. However they form cycles of length $6$ or less, hence it violates the minimum cycle length condition. Finally, we show that $i,j$ must be non leaf nodes in $\mathcal{G}$. For neighbors $i,j$ in $\mathcal{GM}_{hd}$, first consider the case where node $i$ (without loss of generality) is a leaf node in $\cal G$. All neighbors of $i$ in $\mathcal{GM}_{hd}$ are either its parent or neighbors of its parent in $\mathcal{G}$ (separated by two hops). Thus all neighbors $k,l$ of $i$ in $\mathcal{GM}_{hd}$ also have an edge $k-l$ in $\mathcal{GM}_{hd}$ and are not separated by two hops. This contradicts the assumption that $(kl) \notin \mathcal{GM}_{hd}$. Thus $i,j$ must be neighbors and non leaf nodes in $\mathcal{G}$.
\subsection{Proof of Theorem \ref{condind_leaf}}\label{proof_condind_leaf}
We prove the \textbf{if} part using contradiction. Let the true and only neighbor of leaf node $j$ in $\mathcal G$ be node $k\neq i$. Since, $(ij) \in {\cal GM}_{hd}$, it follows that, $i,j$ are two hops away in $\mathcal G$. As there are at least $3$ non-leaf nodes and cycle length is greater than $6$, there exists some node $r$ which is not a leaf node in $\mathcal{G}$ in one of two configuration in $\mathcal G$:\\
(a) $r$ is a neighbor of $i$, two hops away from $k$, and three hops away from $j$.\\
(b) $r$ is a neighbor of $k$, and two hops away from $i$ and $j$.\\
Note that in configuration (a), $r$ is a neighbor of $i$ in $\mathcal G$, but not a neighbor of $j$ in $\mathcal{GM}_{hd}$. In configuration (b), $r$ is a neighbor of $j$ in $\mathcal{GM}_{hd}$, but not $i$'s neighbor in $\mathcal{G}$. This, thus, contradicts the statement of the sets of non-leaf neighbors of $i$ in ${\cal G}$ and $j$ in $\mathcal{GM}_{hd}$ being the same.

For the \textbf{only if} part, consider leaf node $j$ with parent node $i$ in $\mathcal{G}$. As two hop neighbors in $\mathcal{G}$ are neighbors in $\mathcal{GM}_{hd}$, every non-leaf node that is a neighbor of $i$ in $\mathcal{G}$ is a neighbor of $j$ in $\mathcal{GM}_{hd}$ and vice versa.

\subsection{Proof of Theorem \ref{threshld_loopy}}\label{proof_threshld_loopy}
Consider nodes $i,j$ in grid $\mathcal{G}$. Using the expression in Lemma \ref{inverse_LC_PF}, we have
\squeezeup
\begin{align}
 J_{vv} +J_{\theta\theta} &=~\text{\small $H_{g}D^{-1}(\Sigma_{qq}+\Sigma_{pp})H_{g}$}\text{\small$+H_{\beta}D^{-1}(\Sigma_{qq}+\Sigma_{pp})H_{\beta}$}\nonumber
\end{align}
where $D^{-1}(\Sigma_{qq}+\Sigma_{pp})$ is a diagonal matrix with positive entries. Note the structure of $H_g$ and $H_\beta$ from (\ref{Laplacian}). For the \textbf{if} part, consider the case where $i,j$ are neighbors with edge $(ij) \in \mathcal{G}$. As minimum cycle length is greater than three, there are no common neighbors of $i,j$. Hence we have
\begin{align}
 &J_{vv}(i,j) +J_{\theta\theta}(i,j) =\nonumber\\
 &-D^{-1}(i,i)(\Sigma_{qq}(i,i)+\Sigma_{pp}(i,i))(H_g(i,i)g_{ij}+ H_\beta(i,i)\beta_{ij})\nonumber\\
 &-D^{-1}(j,j)(\Sigma_{qq}(j,j)+\Sigma_{pp}(j,j))(H_g(j,j)g_{ij}+ H_\beta(j,j)\beta_{ij})\nonumber\\
 &< 0 \label{neighbor}
\end{align}
To prove the \textbf{only if} part, consider the case where $i,j$ are not neighbors. It follows from the proof of Theorem \ref{structure_LC_PF} that if $i,j$ are three or more hops away, then $J_{vv}(i,j)=0$ and $J_{\theta\theta}(i,j)=0$. Finally, consider the case where $i,j$ are two hop neighbors with common neighbor set $S$. We have
\begin{align}
 J_{vv}(i,j) +J_{\theta\theta}(i,j) &=\sum_{k \in S}D^{-1}(k,k)(\Sigma_{qq}(k,k)+\nonumber\\
 &\Sigma_{pp}(k,k))(g_{ik}g_{kj}+ \beta_{ik}\beta_{kj})
 > 0. \label{twohopneighbor}
\end{align}
Hence proved that the statement holds only if $i,j$ are neighbors.
\end{document}